\documentclass[twocolumn,
superscriptaddress,
amsmath,amssymb,
aps,
prx,
]{revtex4-2}

\usepackage{amsmath,amssymb,amsfonts,dsfont,xspace,graphicx,relsize,bm,mathtools,xcolor,amsthm,soul}
\usepackage{graphicx}
\usepackage{dcolumn}
\usepackage{xcolor}
\usepackage{bbm}
\usepackage{physics}
\usepackage{ulem}
\usepackage{hyperref}
\hypersetup{
	colorlinks = true,
	urlcolor   = blue,
	linkcolor  = blue,
	citecolor  = red
}
\usepackage{algorithm}
\usepackage{algpseudocode}
\usepackage{float}
\usepackage{siunitx}
\usepackage[shortlabels]{enumitem}

\newcommand{\ave}[1]{\left\langle#1\right\rangle}
\newcommand{\1}{\mathbbm{1}}

\newcommand{\bX}{{\bf X}}
\newcommand{\bY}{{\bf Y}}
\newcommand{\bZ}{{\bf Z}}
\newcommand{\bP}{{\bf P}}

\newcommand{\cH}{{\mathcal{H}}}
\newcommand{\D}{{\mathcal{D}}}
\newcommand{\E}{{\mathbb{E}}}

\newcommand{\X}{{\mathsf{X}}}

\newcommand{\ii}{{\rm i}}
\makeatletter
\let\X\@undefined
\makeatother
\newcommand{\X}{{\mathcal{X}}}

\newcommand{\dketbra}[1]{|{#1}\rangle\langle{#1}|}
\renewcommand{\tr}[1]{{\rm Tr}\left[#1\right]}

\newcommand{\locc}{{\rm LOCC}_1}

\newcommand{\beq}{\begin{equation}}
\newcommand{\beql}[1]{\begin{equation}\label{#1}}
\newcommand{\eeq}{\end{equation}}
\newcommand{\eeqp}{\,\,\,.\end{equation}}
\newcommand{\eeqc}{\,\,\,,\end{equation}}

\theoremstyle{plain}

\newtheorem{definition}{Definition}

\newtheorem{theorem}{Theorem}

\newtheorem{lemma}{Lemma}

\makeatletter
\renewcommand{\ALG@name}{Protocol}
\newenvironment{protocol}[2][Default Title]
  {\begin{algorithm}[H]  
   \caption{#1}\label{#2}%
   \begin{algorithmic}[1]
   \raggedright}%
  {\end{algorithmic}
   \end{algorithm}}
\makeatother

\makeatletter
\let\newfloat\newfloat@ltx
\makeatother

\begin{document}
\title{Faithful and secure distributed quantum sensing under general-coherent attacks}
\author{G. Bizzarri}
\affiliation{Dipartimento di Scienze, Universit\`a degli Studi Roma Tre, Via della Vasca Navale, 84, 00146 Rome, Italy}
\author{M. Barbieri}
\affiliation{Dipartimento di Scienze, Universit\`a degli Studi Roma Tre, Via della Vasca Navale, 84, 00146 Rome, Italy}
\affiliation{Istituto Nazionale di Ottica - CNR, Largo E. Fermi 6, 50125 Florence, Italy}
\author{M. Manrique}
\affiliation{Dipartimento di Scienze, Universit\`a degli Studi Roma Tre, Via della Vasca Navale, 84, 00146 Rome, Italy}
\author{M. Parisi}
\affiliation{Dipartimento di Scienze, Universit\`a degli Studi Roma Tre, Via della Vasca Navale, 84, 00146 Rome, Italy}
\author{F. Bruni}
\affiliation{Dipartimento di Scienze, Universit\`a degli Studi Roma Tre, Via della Vasca Navale, 84, 00146 Rome, Italy}
\author{I. Gianani}
\affiliation{Dipartimento di Scienze, Universit\`a degli Studi Roma Tre, Via della Vasca Navale, 84, 00146 Rome, Italy}
\author{M. Rosati}
\affiliation{Dipartimento di Ingegneria Civile, Informatica e delle Tecnologie Aeronautiche, Universit\`a degli Studi Roma Tre, Via Vito Volterra 62, 00146 Rome, Italy}
\begin{abstract}
Quantum metrology and cryptography can be combined in a distributed and/or remote sensing setting, where distant end-users with limited quantum capabilities can employ quantum states, transmitted by a quantum-powerful provider via a quantum network, to perform quantum-enhanced parameter estimation in a private fashion. Previous works on the subject have been limited by restricted assumptions on the capabilities of a potential eavesdropper and the use of abort-based protocols that prevent a simple practical realization. Here we introduce, theoretically analyze, and experimentally demonstrate single- and two-way protocols for distributed sensing combining several unique and desirable features: (i) a safety-threshold mechanism that allows the protocol to proceed in low-noise cases and quantifying the potential tampering with respect to the ideal estimation procedure, effectively paving the way for wide-spread practical realizations; (ii) equivalence of entanglement-based and mutually-unbiased-bases-based formulations; (iii) robustness against collective attacks via a LOCC-de-Finetti theorem, for the first time to our knowledge. Finally, we demonstrate our protocols in a photonic-based implementation, observing that the possibility of guaranteeing a safety threshold may come at a significant price in terms of the estimation bias, potentially overestimating the effect of tampering in practical settings. 
\end{abstract}
\maketitle
\section {Introduction} 
Quantum metrology enables the estimation of physical quantities with precision surpassing classical limits, thanks to the use of coherence and entanglement at the stages of probe preparation and measurement~\cite{Giovannetti2011}. In recent years, advancements in the manipulation and transfer of quantum states over long distances, along with the vision of a fully connected quantum internet, have kindled the study of quantum metrology in a distributed and/or remote setting~\cite{Yin2020,Takeuchi2019,Moore2023,Huang2019a,Shettell2022a,Alan_Rogers_1999,QI2001655,9354971,5955066,PhysRevLett.120.080501,PhysRevLett.121.043604,PhysRevA.97.032329,Kasai2022,expDQS,Zhao2021a,DQSclocks,He2024,Xie18,Ho2024,Hassani_2025,Shettell2022,Bugalho2024,Zang2024,Ho2024}. Here, one or more distant parties aim to estimate some parameters using quantum states transmitted via a public quantum network. Therefore, on top of the enhanced estimation capabilities provided by quantum sensors, one is typically interested in guaranteeing some notion of privacy with respect to other legitimate users and security against potential eavesdroppers.  
Such scenario applies to a variety of practical settings for near- and mid-term quantum technologies, characterized by the presence of a fully quantum service-provider and an end-user client with limited quantum capabilities, e.g., the sensing of private medical data or of a natural phenomenon happening at a remote location.

\begin{figure}[b!]
    \centering
    \includegraphics[width=\columnwidth]{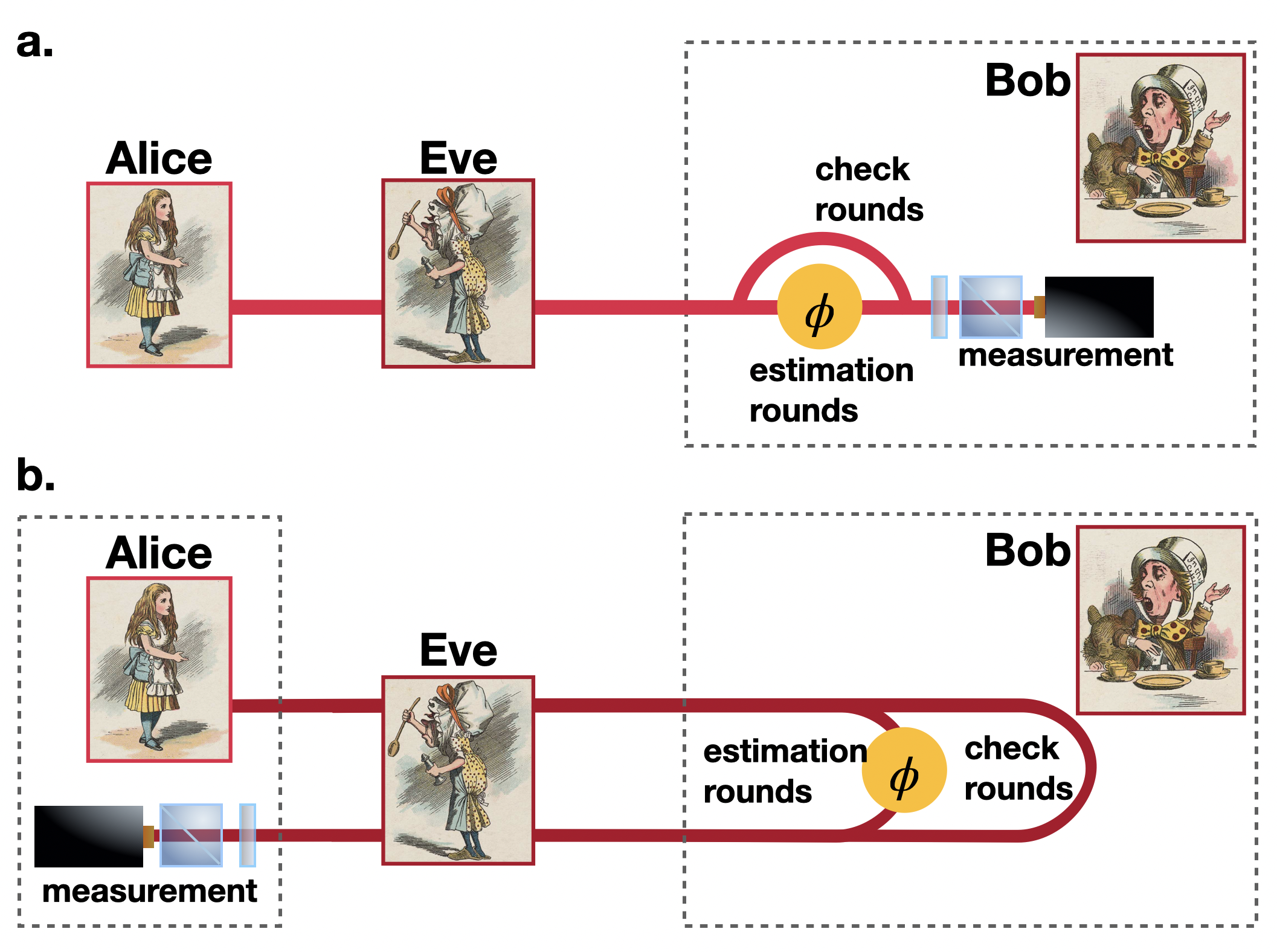}
    \caption{Depiction of the general DQS setting we consider. A provider Alice prepares quantum probe states and sends them to Bob, in order to estimate a parameter $\phi$  at a remote location. The transmission takes place on a quantum channel potentially controlled by a malicious adversary Eve, who can try to: (i) tamper with the probe in order to compromise the estimation; (ii) leak information in order to estimate the phase herself. (a) In the single-way setting, Bob can perform a quantum measurement at his location. (b) In the two-way setting, Bob is a passive party with no measurement capabilities, hence the quantum state has to travel back to Alice after the parameter-encoding stage.}  
    \label{fig:scheme}
\end{figure}

In particular, restricting to the single-user setting, previous works have focused on the security requirement, relying on similar arguments to quantum key distribution (QKD):~\cite{Yin2020,Takeuchi2019} proposed an entanglement-based protocol, 
where the eavesdropper, Eve, is classical, and they can only intercept the measurement outcomes of the estimation procedure; furthermore, in~\cite{Moore2023}, the authors considered the more complex scenario where Eve can intercept both quantum and classical communication, while the provider Alice wants to estimate a phase at a remote location controlled by a honest party Bob, without anyone else acquiring information on the phase itself. 
The general setting, instead, was studied in~\cite{Huang2019a,Shettell2022a}, where the authors introduced single- and two-way estimation protocols between two legitimate users Alice and Bob, that communicate through quantum and classical channels controlled by Eve. However, a key limiting factor towards the implementation of these general protocols is that security can be enforced only at the cost of aborting the estimation procedure as soon as an error is detected. This is achieved by sending decoy or flag states on the quantum channel, and detecting whether they have been altered after transmission; in practice, such a zero-tolerance policy can lead to aborting the protocol at every round due to inevitable experimental errors, such as imperfect measurements or dark counts. 
Finally, the security of the above protocols has been analyzed only under the specific assumption of individual attacks by Eve, which, as in QKD, leaves it open to attacks that employ a quantum memory to establish quantum correlations between probes transmitted during subsequent uses of the channel. 

In this article, we take a unifying approach to distributed quantum sensing (DQS), tackling the above-mentioned issues and providing experimental-friendly protocols for the estimation of a phase-shift assisted by public quantum and classical communication. The protocols are analyzed under the following specific requirements:
\begin{enumerate}
    \item \textit{Faithfulness} The legitimate users can estimate how much Eve has tampered with the ideal probe state, and hence quantify how much the actual estimation bias and error deviate from the ideal ones. Note that this is a combination of soundness and integrity as defined by~\cite{Shettell2022a}.
    \item \textit{Security} The legitimate users can quantify how much information about the estimated parameter has leaked to Eve. 
\end{enumerate}
We present both a single-way protocol, where the probe remains in Bob's lab after phase-encoding, and a two-way protocol, where the probe is sent back to Alice for measurement, discussing how security cannot be guaranteed in the second case, while faithfulness is preserved. We note that our protocols can be equivalently formulated in terms of entanglement or mutually unbiased bases (MUBs), thus bridging previous approaches to the problem. Furthermore, for the first time to our knowledge, we prove that the protocols' properties are robust with respect to collective attacks by Eve, making use of a de Finetti theorem under single-way adaptive measurements due to~\cite{Li2015}. 
Finally, we test our protocols via the implementation of optical-phase estimation using the photon polarization degree of freedom, providing evidence that the faithfulness can be guaranteed, in practice, only at the price of a potentially large bias on the estimated quantity. 

The rest of the article is structured as follows: section II reports on related results, section III provides the theoretical framework, section IV illustrates a proof-of-principle experiment, and finally we conclude with discussion and outlook in section V. The main supporting  methods are presented in section VI, while more technical aspects are referred to the Appendices.

\begin{centering}
\begin{table*}[!]
\centering    \begin{tabular}{|>{\centering\arraybackslash}p{0.17\textwidth}||>{\centering\arraybackslash}p{0.14\textwidth}|>{\centering\arraybackslash}p{0.14\textwidth}|>{\centering\arraybackslash}p{0.14\textwidth}|>{\centering\arraybackslash}p{0.14\textwidth}|>{\centering\arraybackslash}p{0.15\textwidth}|} \hline 
        DQS Protocol &   Quantum communication&Attacks&    Security&Faithfulness&Entanglement-based\\ \hline \hline
         Protocol~\ref{protocol:1w_ent}&   1- \& 2-way&General-coherent&    Perfect (1w), None (2w)&Threshold-based&Yes\\ \hline 
         Protocol~\ref{protocol:1w_mub}&   1- \& 2-way&General-coherent&    Perfect (1w), None (2w)&Threshold-based&No\\ \hline 
         \cite{Yin2020,Takeuchi2019}&   1-way&Classical&    Perfect&A posteriori&Yes\\ \hline
 \cite{Moore2023,He2024}&  1-way&Individual&   Perfect&A posteriori&No\\ \hline
 \cite{Huang2019a}&  1- \& 2-way&Individual&   Perfect (1w), None (2w)&Abort-based&No\\\hline
 \cite{Shettell2022}&  1- \& 2-way&Individual&   Perfect (1w), Zero-tolerance (2w) &Abort-based&No (trap code), Yes (Clifford code)\\\hline
    \end{tabular}
    \caption{Comparison of different single-user DQS protocols in the literature with the ones proposed in this article. Attacks are labelled according to the capabilities of the eavesdropper (Eve); the security is deemed perfect whenever Eve holds a phase-independent completely mixed state, while it is zero-tolerance when based on passing a check. The faithfulness is labelled threshold-based when there is a safety margin for errors detected during the protocol, while it is abort-based when the margin is zero; instead, protocols working a posteriori require a complete tomographic analysis of the received state. }
    \label{tab:my_label}
\end{table*}
\end{centering}

\section{Related works}
Distributed sensing is routinely used in current technologies as a way of monitoring systems, yielding information on their local and global properties as well~\cite{Alan_Rogers_1999,QI2001655,9354971,5955066}. This has prompted the investigation of its extension to quantum sensing~\cite{PhysRevLett.120.080501}, requiring considerations on the optimal distribution of resources~\cite{PhysRevLett.121.043604,PhysRevA.97.032329}. It has been shown in an experiment that sharing entanglement among different partners estimating the average phase in their network is preferable to using local squeezing~\cite{expDQS}, followed by demonstrations with polarisation-entangled states~\cite{Zhao2021a,Kim24}.  In the same vein, adopting a network of atomic clocks linked by entanglement has been identified as a promising platform to realize a quantum-limited global positioning system~\cite{DQSclocks}. 

Following on the integration of quantum metrology with communication, one is naturally led to consider the presence of malevolent agents in the network, being them in the form of eavesdroppers or untrusted nodes. This question has been considered in the experiment of Ref. \cite{Yin2020} based on the proposal in~\cite{Takeuchi2019}. The authors have implemented a secure link between two parties by means of a polarisation-entangled state, as in the scenario of Fig.~ \ref{fig:scheme}. The integrity of the communication link has been verified by quantum state tomography, and, from the knowledge of the state, it has been possible to derive a fiducial statistical model for phase estimation. Crucially, this has shown a large asymmetry in the Fisher information available to the two parties and that potentially accessible to an eavesdropper. In addition, since the receiving party has no information on Alice's control, the protocol is also oblivious to Bob. The work~\cite{Moore2023} has extended this protocol by relaxing the requirement of utilising an entangled state. This can be replaced with random separable states 
sent by Alice. Its experimental realisation in~\cite{He2024} has shown the same asymmetry in Fisher information, with a setup compatible with quantum communication standards. In our article, we show that these two kinds of protocols are equivalent. Furthermore, rather than studying the estimation performance of Eve and Alice in a controlled scenario, we demonstrate Alice's capability of quantifying how much Eve has modified the quantum probe, and therefore setting up a safety threshold for aborting the protocol. 

Ref.~\cite{Huang2019a} has introduced fully cryptographic methods in quantum estimation with qubits, shortly followed by a generalisation to high-dimensional systems in~\cite{Xie18}. The basic idea consists in extending the set of probe states to include decoy events along those that effectively estimate the parameter. This work has introduced the distinction between one-way protocols, in which Bob performs the measurement as well as imparting the parameter, and two-way protocols, in which Bob finally sends his state back to Alice following some possible information scrambling. This cryptographic approach has then been extended in~\cite{Shettell2022a}, enlarging the class of attacks to which the original proposal~\cite{Huang2019a} is resilient. This comes at the cost of considering a large joint state of probe and flag states, which are then encrypted by means of code, in the form of Clifford operations, either collective or separable. This work has also the merit of introducing clear definitions for privacy, soundness, and integrity. Instead, in our article, we strive to keep the protocol simple in terms of state-space dimension, while enlarging the class of attacks to the most general one allowed in QKD, i.e., where Eve can perform collective operations across different rounds of the protocol. This requires a significant extension of the theoretical methods introduced in the literature. 

An interesting application has been presented in~\cite{Shettell2022}, followed by its experimental realisation in~\cite{Ho2024}. This scenario considers the estimation of a global phase parameter in a network, preventing information leakage about the local phase at the individual nodes. The notion of privacy in such quantum sensor networks has been elaborated in Ref.~\cite{Hassani_2025}. Here we restrict to a single-provider/single-user setting, but we expect that our methods can be extended to the multi-user scenario.

A summary of and a comparison with previous works related with our results on single-user DQS is provided in Table~\ref{tab:my_label}.

\section{Theoretical framework}
\subsection{Setting and notation}
Our DQS setting is shown in Fig.~\ref{fig:scheme}: a provider Alice with fully quantum capabilities can exchange quantum states $\ket\psi$ on a quantum channel with a technologically-limited party Bob, placed at a remote sensing location where a parameter $\phi$ can be encoded via a unitary operation $U(\phi)$. Bob's capabilities can range from a semi-passive switch, which essentially turns on or off to determine when the phase gets encoded, and then simply sends the probe quantum state back to Alice, to a quantum-classical sensor that can perform a destructive measurement of the probe and announce the outcome to Alice via a classical channel. Any kind of classical or quantum communication is assumed to be public, and hence controlled by a malicious adversary, Eve, with two distinct aims: 
\begin{enumerate}
    \item \textit{Tampering attack:} disrupting the estimation procedure by modifying the states transmitted on the quantum channel;
    \item \textit{Leaking attack:} extracting information about the estimation procedure, in order to estimate the value of the parameter.
\end{enumerate}
The performance of the protocol is measured in response to these two kinds of attacks: \textit{faithfulness} quantifies the resilience to tampering attacks, and it can be measured by the distance of the ideal state from the tampered state after phase encoding, which bounds the estimation quality~\cite{Shettell2022a}; \textit{security} quantifies the resilience to leaking attacks, and it can be measured by the distance of Eve's share of the quantum state from the maximally mixed one, which allows no estimation at all. 
At variance with previous works on the subject~\cite{Huang2019a,Shettell2022a}, we make no assumptions on Eve's capabilities, hence include the possibility to entangle the transmitted probe states with ancillary ones across multiple repetitions of the protocols and store the ancillae in a quantum memory for later collective quantum processing; this is the largest possible set of causal attacks by a quantum adversary, called general-coherent (GC) attacks in QKD~\cite{Pirandola2019}. 

In the following, without loss of generality, we assume that the probe state is that of an $n$-qubit system, and the parameter-encoding unitary is of the form $U(\phi)^{\otimes n}$ with $ U(\phi) = e^{\ii \phi Y}$, with $Y=\dketbra{R}-\dketbra{L}$ the second Pauli matrix, $\ket{R/L} = (\ket{0}\pm i \ket{1})/\sqrt{2}$ in the computational basis, and $\phi$ the phase to be estimated. It is well-known that an optimal probe for estimating $\phi$ is the entangled state
\begin{equation}
    \ket{\boldsymbol+} = \frac{\ket{\bf 0} + \ket{\bf1}}{\sqrt2},
\end{equation}
where $\ket{\bf i}=\ket{i}^{\otimes n}$ for $i=0,1$, yielding a factor-of-$n$ reduction in the variance with respect to a product state~\cite{Giovannetti2011}.
For later use, we indicate with bold symbols $\bX, \bY, \bZ$ the equivalent of Pauli matrices in the subspace ${\rm span}\{\ket{\bf 0},\ket{\bf 1}\}$ of the $n$-qubit space, in order to distinguish them from the ordinary Pauli matrices $X, Y, Z$.
The fidelity between two quantum states is $F(\rho,\sigma) = ||\sqrt\rho\sqrt\sigma||_1^2$, while the trace-distance is $D(\rho,\sigma) = \frac12||\rho-\sigma||_1$, with $||\cdot||_1$ the Schatten $1$-norm; when $\rho$ and/or $\sigma$ are pure states, we indicate them with the corresponding ket in the arguments of these functions. We indicate with $||\cdot||_\infty$ the operator norm. Where explicitly stated, ${\cal D}(\cH)$ is the space of density matrices on the Hilbert space $\cH$. For a state $\rho\in{\cal D}(\cH^{\otimes(n + m)})$, we indicate as $\Tr_m[\rho]\in{\cal D}(\cH)$ the reduced state obtained by tracing out $m$ parties. 


\onecolumngrid
\begin{centering} 
\begin{minipage}{\textwidth}
\begin{protocol}[One-way entanglement-based DQS]{protocol:1w_ent}
\item\label{step:preparation} \textbf{Preparation:}\\
Alice prepares the entangled state
\begin{equation}
\label{eq:resourcestate}
    \ket{\boldsymbol{\psi_+}}_{AB} = \frac{\ket{0}_A \ket{\bf0}_B + \ket{1}_A \ket{\bf 1}_B}{\sqrt2},
\end{equation}
where $A$ is a single-qubit and $B$ is an $n$-qubit system. \\
She then sends the $B$ register to Bob via the quantum channel. \\
\item \textbf{Encoding:}\label{step:encoding}\\
Bob chooses to perform one of these three actions at random:
\begin{enumerate}[(a)]
    \item with probability $p_c$ he does not encode the phase;\\
    \item with probability $p_e$, he encodes the phase by applying $U(\phi)$ to the received state;\\
    \item with probability $p_d = 1-p_c-p_e$, he discards the received state and moves on to the next round of the protocol.\\
\end{enumerate}
\item\label{step:measurement} \textbf{Measurement:}\\
Alice measures one of the observables $X, Y, Z$ chosen uniformly at random.\\
Bob performs on of these two actions, depending on what he did at step~\ref{step:encoding}:
\begin{enumerate}[(a)]
    \item if he did not encode the phase, he measures one of the observables $\bX, \bY, \bZ$, chosen uniformly at random;\\
    \item if he encoded the phase, then he measures one of the observables $\bX, \bZ$, chosen uniformly at random;\\
\end{enumerate}
\item \textbf{Reconciliation and sifting:}\label{step:reconciliation} \\
After $T$ repetitions of steps \ref{step:preparation}-\ref{step:measurement}, Alice communicates to Bob the observables she measured in each round and the corresponding outcomes. \\
Bob then keeps all rounds in which the following products of observables were measured, discarding the others: 
\begin{enumerate}[(a)]
    \item $X_A \otimes \bZ_B$, $Z_A \otimes \bX_B$, and $Y_A \otimes \bY_B$, if he did not encode the phase;\\
    \item $X_A \otimes \bZ_B$ and $Z_A \otimes \bX_B$, if he encoded the phase.\\
\end{enumerate}
\item \textbf{Fidelity check:}\\
From the sifted rounds (a) where he did not encode the phase, Bob estimates the fidelity of the received state with the ideal state $\ket{\boldsymbol{\psi_+}}$ as follows:
\begin{equation}\label{eq:fidelity_ent}
    \hat F = \frac{1+\ave{X\otimes \bZ}+\ave{Z\otimes \bX}+\ave{Y\otimes \bY}}{4},
\end{equation}
where the expectation is taken with respect to the tampered state.
If $\hat F \geq 1-\epsilon^2$, for a given safety threshold $\epsilon$, then Bob proceeds to the next step \ref{step:estimation}. \\
Otherwise, he aborts the protocol.
\item\label{step:estimation} \textbf{Phase estimation:}\\
From the sifted rounds (b) where he applied $U(\phi)$, Bob estimates the phase as
\begin{equation}
    \frac12\arccos\left(\frac{\ave{X\otimes \bZ}+\ave{Z\otimes \bX}}{2}\right),
    \label{eq:phase_estimator}
\end{equation}
where the expectation is taken with respect to the phase-encoded tampered state.
\end{protocol}
\end{minipage}
\end{centering}
\twocolumngrid

\subsection{DQS protocol}
The main protocol that we will consider is described in Protocol~\ref{protocol:1w_ent} (see also Protocol~\ref{protocol:1w_mub} for an equivalent version based on MUBs). It relies on shared entanglement between the provider, Alice, and the end-user, Bob, with a randomization of phase-encoding and measurements, which enables them to detect and quantify potential tampering from Eve. Furthermore, it is formulated directly in terms of measurements of Pauli-like observables to facilitate experimental realization. Indeed, while Alice keeps a single qubit in her lab, she can send an $n$-qubit probe state to Bob to achieve quantum-enhanced estimation performance. Nevertheless, Bob only needs to measure two-dimensional observables on the $n$-qubit space. 

Note that steps \ref{step:reconciliation}-\ref{step:estimation} can also work in a reverse-reconciliation fashion, i.e., with Bob announcing his measurements, outcomes and check/estimation rounds to Alice, since Eve does not have access to the observables measured by Alice in each round. However, security would be immediately compromised in that case, unless Bob first announces the results of check rounds and, once the check is passed and they are sure that Eve has meddled little with the probe, he announces the results of the estimation rounds. 

The parameters $p_e$, $p_c$ and $p_d$, together with the measurements' randomization, determine the fraction of check, estimation and discarded rounds after sifting: for large $T$, we expect to obtain $N_c \simeq p_e T /3$ sifted check rounds and $N_e \simeq p_c T/3$ sifted estimation rounds, while approximately $N_d \simeq p_d T + (1-p_d) T/3 = (1+2p_d)T/3$ rounds are discarded. These latter rounds are need to ensure that the protocol is robust to GC attacks, satisfying the requirement $N_d\gg N_e + N_c$ via the choice of a suitable $p_d\gg0$. If, instead, one restricts to individual attacks, it is sufficient to take $p_d=0$. 

Furthermore, note that Protocol~\ref{protocol:1w_ent} can be easily extended to a two-way framework, wherein Bob takes on a more passive role: in step~\ref{step:encoding}, he can only decide between actions (a) or (b), after which he sends the $n$-qubit quantum state back to Alice via the quantum channel. The subsequent steps are then all carried out by Alice in her lab, including case (c), i.e., discarding a fraction $p_d$ of the received states. Unfortunately, in this case only faithfulness can be guaranteed, while security is rendered impossible by Eve's interaction with the probe before and after phase-encoding, which can potentially give her full access to the phase information without Alice noticing. Similar observations in~\cite{Shettell2022a} led the authors to devise an encryption-decryption mechanism for Bob's phase-encoding, which significantly complicates the protocol, yielding itself to a large abortion rate in practical noisy scenarios. 

The faithfulness and security guarantees of our proposed protocol are provided in the following Theorem:
\begin{theorem}    \label{thm:main}
The one-way entanglement-based DQS Protocol~\ref{protocol:1w_ent} is perfectly secure and its faithfulness is determined by $\epsilon$. Indeed, the difference of bias and variance of an unbiased estimator $\hat\phi$ calculated on the ideal probe state with respect to the same estimator $\hat\phi'$ on the tampered state are bounded, under general-coherent attacks, as
\begin{align}
    &\left|\E \hat\phi-\E\hat \phi'\right| \leq \frac{\epsilon_0}{n |\sin(2n\phi)|}, \label{eq:bias_gc}\\
    & \left|\Delta^2\hat\phi-\Delta^2 \hat\phi'\right| \leq \frac{(2\epsilon_0+\epsilon_0^2)}{n^2 \sin^2(2n\phi)}.\label{eq:var_gc}
\end{align}
with $\epsilon_0 = \sqrt{\frac23\epsilon^2 + 4 f\left(T,N_d,n\right)}$ and $f(T,x,n) = (T-x-1)\sqrt{\frac{n}{2x}}$.
In the case of individual attacks, one can take $N_d=0$ and employ the same bias bound of  \eqref{eq:bias_gc}  with $\epsilon_0\mapsto\sqrt\frac23\epsilon$, while the variance bound reads
\begin{equation}
    \left|\Delta^2\hat\phi-\Delta^2 \hat\phi'\right| \leq \frac{ (2\epsilon_0/N_e+\epsilon_0^2)}{n^2 \sin^2(2n\phi)}.
\end{equation}
The two-way version of the protocol has similar bounds with $\epsilon_0\mapsto \sqrt{\frac23\epsilon^2 + 4 f\left(T,N_d,n\right)} + |\sin(\min\left\{n\phi,\pi/2\right\})|$.
\end{theorem}
The proof of Theorem~\ref{thm:main}, deferred to Section~\ref{sec:methods}, relies on three main steps: (i) showing the equivalence of Protocol~\ref{protocol:1w_ent} to a MUB-based protocol where Alice sends as probe state the eigenstate of a random Pauli matrix and Bob applies corresponding measurements during the check and estimation rounds; (ii) showing that, even in the presence of GC attacks,  the fidelity measured during check rounds quantifies equally well the fidelity to the phase-encoded state used in the estimation rounds, via a suitable de-Finetti theorem; (iii) correcting the bias and variance estimates accordingly.

\section{Experimental implementation and results}
Protocol~\ref{protocol:1w_ent} was tested at the proof-of-principle stage. We did not aim at reproducing all features of a complete  implementation, but rather to obtain experimental data compatible with it. For this purpose, we employed photonic qubits with polarisation encoding by means of the setup presented in Fig.~\ref{fig:exp_setup}. A photon pair is produced by spontaneous parametric down-conversion, and delivered on a photonic control-Z gate in order to generate a state close to  maximal entanglement - this correponds to setting $n=1$ in Alice and Bob's shared resource Eq.~\eqref{eq:resourcestate}.
\begin{figure}
    \centering
    \includegraphics[width=\linewidth]{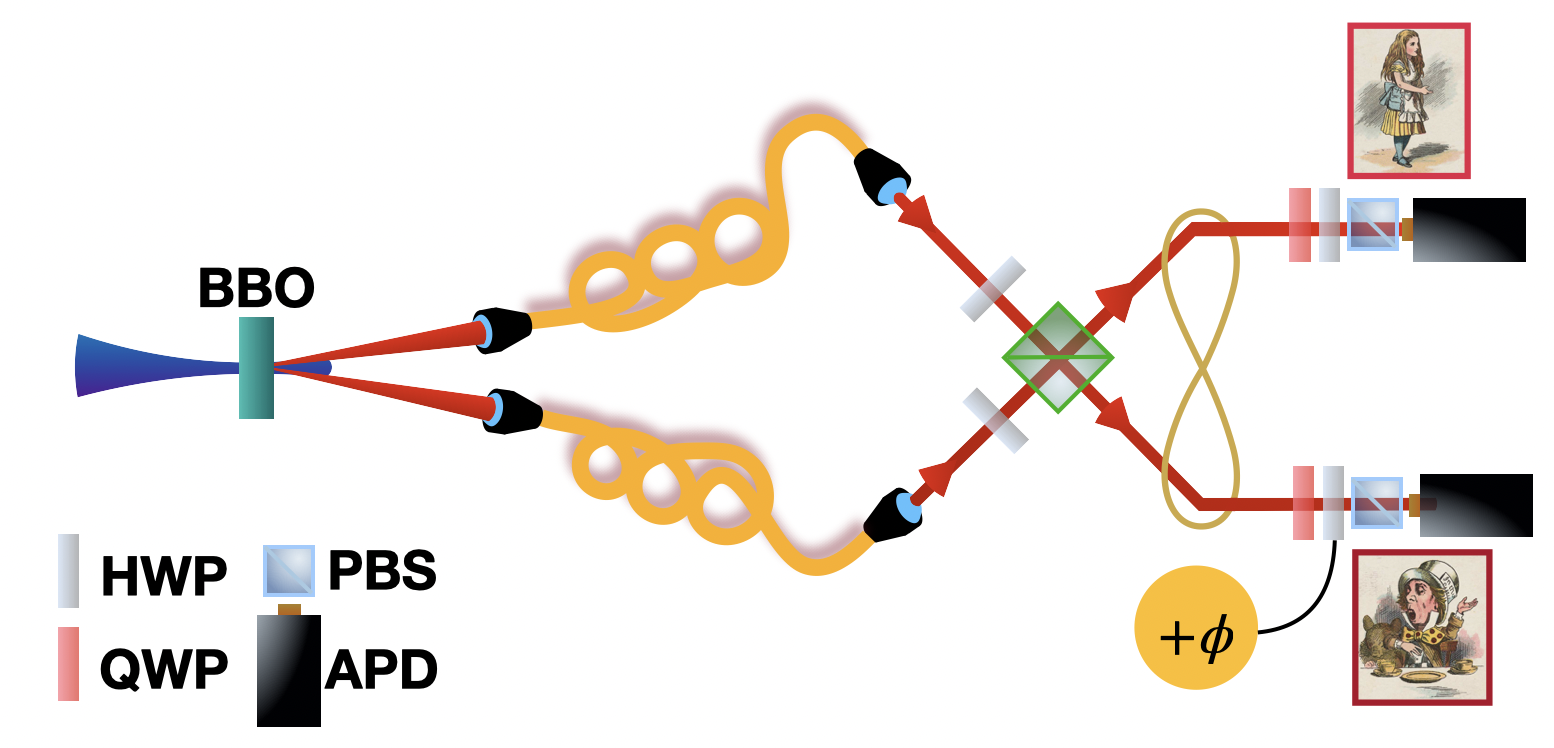}
    \caption{Sketch of the experimental setup: photon pairs generated by SPDC process in a BBO crystal are sent via single-mode optical fibers into a control-Z gate. After the two photons are entangled, their polarization is measured. Here, on estimation rounds, the target phase $\phi$ is imparted on one arm. Detected photons are then delivered to avalanche photodiodes (APDs) single photon detectors via single-mode optical fibers.}
    \label{fig:exp_setup}
\end{figure}
Polarisation measurements, associated to Pauli operators, are carried out as customary, by means of waveplate and polarisers. During the check rounds, the correlators $\ave{X\otimes Z},\ave{Z\otimes X}, \ave{Y\otimes Y}$ are reconstructed, leading to a fidelity $ F = 0.937\pm0.017$ ($\epsilon=0.251\pm0.034$), calculated according to \eqref{eq:fidelity_ent}. We attribute all of the discrepancy to the tampering action of Eve, even if we do not perform attacks actively. As for the estimation rounds, the phase is encoded not by means of an additional object, but by an additional rotation of Bob's plate by an angle $\theta$ - this corresponds to a phase $\phi=2\theta$. The estimator \eqref{eq:phase_estimator} valid for the ideal state has been employed: its bias \eqref{eq:bias_gc} and discrepancy for the variance \eqref{eq:var_gc} have then been studied as a function of $\theta$. The results are presented in Figs.~\ref{fig:dati} and \ref{fig:bias}.

\begin{figure}[h!]
    \centering
    \includegraphics[width= \columnwidth]{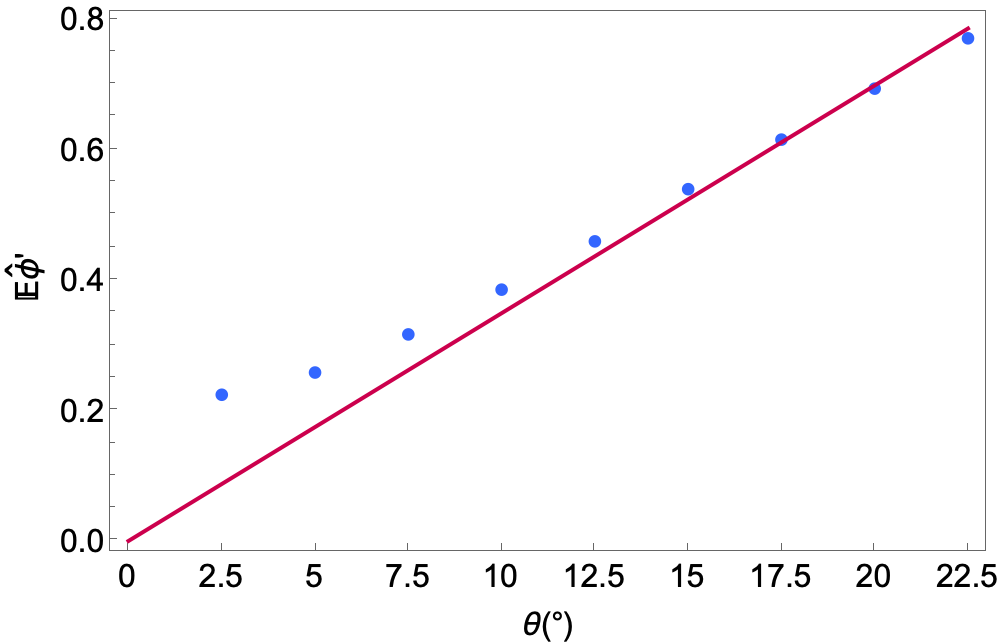}
    \includegraphics[width=\columnwidth]{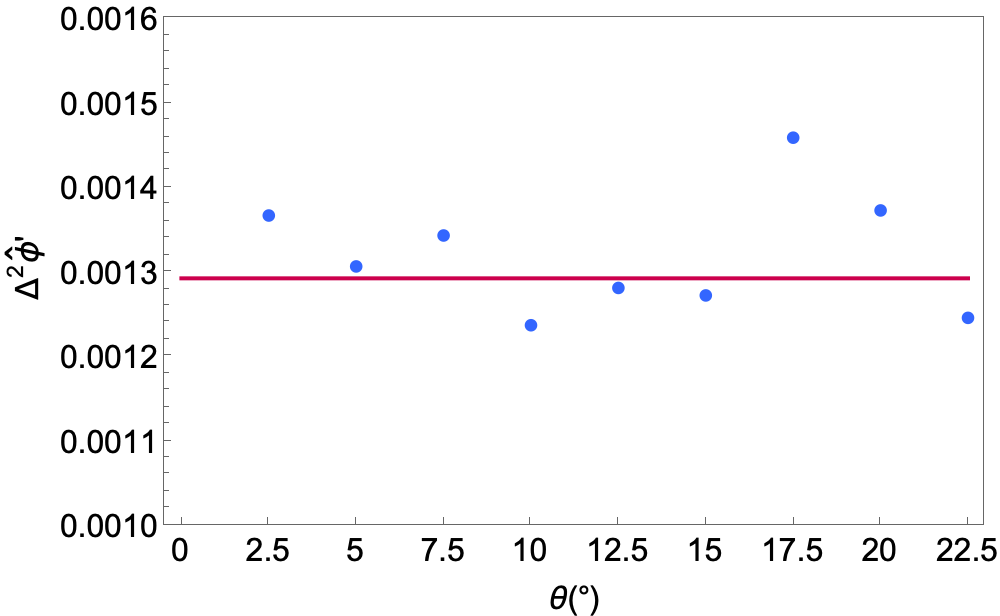}
    \caption{Results of the phase estimation. Top: retrieved value $\mathbb{E} \hat\phi'$ (blue points) as a function of the plate setting $\theta$, compared with the expectation $\phi'=2\theta$ (solid red curve). Bottom: experimental value $\Delta^2 \hat\phi'$ (blue points), compared with the Cram\'er-Rao bound for the ideal maximally entangled state (solid red curve). }
    \label{fig:dati}
\end{figure}

\begin{figure}[h!]
    \centering
    \includegraphics[width=\columnwidth]{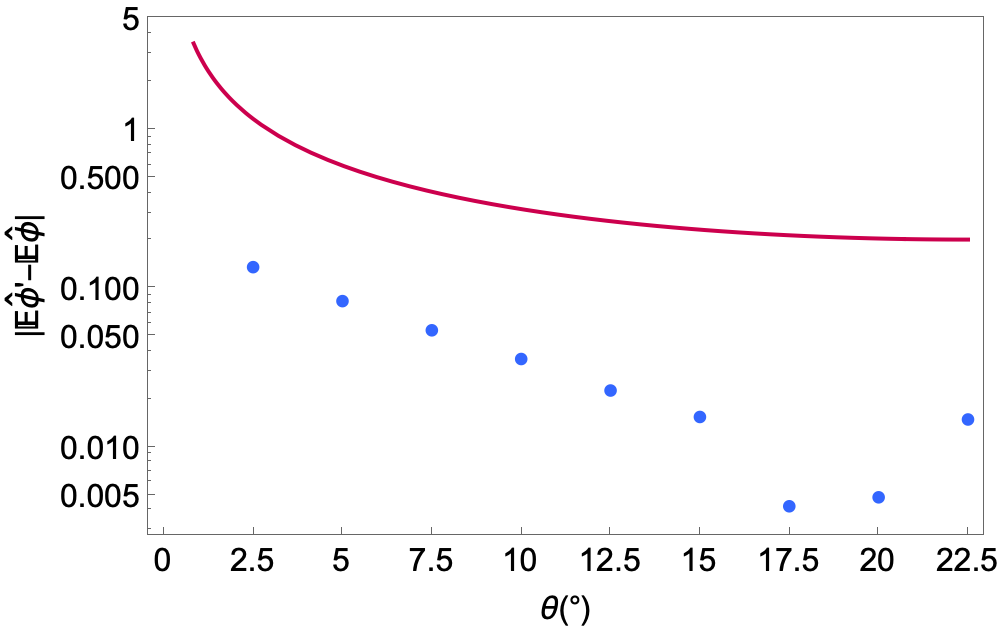}
    \includegraphics[width=\columnwidth]{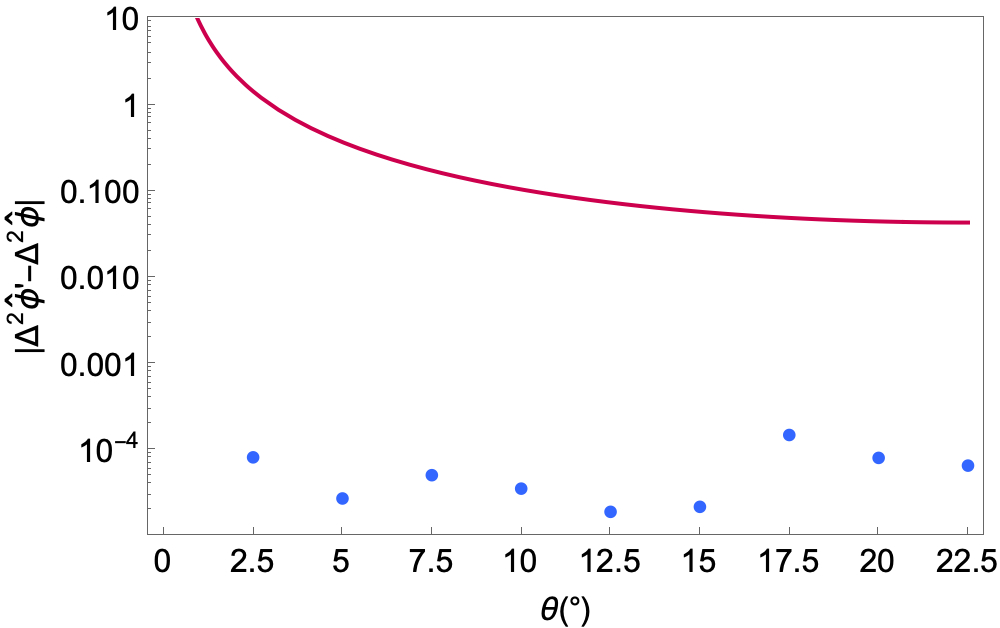}
    \caption{Precision of phase estimation. Top: bias with respect to the ideal case (blue points), compared with the bound \eqref{eq:bias_gc} calculated based on the experimental fidelity (red solid line). Bottom: discrepancy from the ideal case (blue points), compared with the bound \eqref{eq:var_gc} calculated based on the experimental fidelity and the available events (red solid line).}
    \label{fig:bias}
\end{figure}


From Fig.~\ref{fig:dati} we observe that the impact of tampering on phase estimation is more pronounced for small values of $\phi$. However, this is not reflected in the variance, which shows a scattered profile in proximity of the ideal Cram\'er-Rao limit.  Instead, Fig.~\ref{fig:bias} clearly shows that the deviation of bias and variance predicted theoretically from Theorem~\ref{thm:main} tend to significantly overestimate the real ones.  

\section {Discussion and outlook} 
In this article we have proposed a DQS protocol with several desirable features, including: (i) practical security criterion via a threshold-based approach; (ii) faithfulness of the estimation guaranteed against the largest class of collective quantum attacks, i.e., general-coherent attacks; (iii) ease of implementation with minimal resources, via bipartite entangled states or mutually unbiased bases. Our work tackles multiple issues raised in previous works, providing a unified information-theoretic perspective on remote sensing with cryptographic guarantees. We expect that our methods will be applicable to the multi-user scenario, complementing the purely metrological approach of the literature~\cite{Shettell2022,Shettell2022b,Hassani_2025}, as well as to continuous-variable quantum systems where the combination of communication and sensing problems is common~\cite{Bilkis2021,Rosati2021,Rosati2023,Notzel2022,Munar-Vallespir2024}. 

From a practical perspective, our results highlight that the bound on bias and variance with respect to the ideal estimation can significantly overestimate Eve's effective tampering; we attribute this overestimation to the use of trace-distance, which can be easily employed to bound the deviation of bias and variance in theory, while it must be estimated via the fidelity in practice. On one hand, this leaves open the possibility of finding more refined bounds that take into account the actual observable being measured, seeking to employ the measured fidelity directly. On the other hand, we cannot exclude that, without acquiring more specific information on the transmitted probe state, the quantification of metrological performance exclusively via cryptographic parameters (e,g. faithfulness) is fundamentally limited. Future works on the subject should address this conundrum.

Finally, tackling the tradeoff between security, faithfulness, and the protocol's complexity for guaranteeing them might require a thoughtful and accurate characterization of Bob's technological capabilities, also accounting for the distributed and/or remote problem setting. For example, in order to ensure security in the two-way scenario, it might be strictly necessary to require that Bob can create entanglement such as in Ref.~\cite{Shettell2022a}, potentially rendering remote state generation unnecessary. It remains an open question whether other strategies, based on lesser technological demands on Bob's side, might still grant security.

\section{Methods}\label{sec:methods}

\subsection{Theoretical methods}
\subsubsection{MUB-based protocol}
The first step towards proving Theorem~\ref{thm:main} is reducing it to a MUB-based protocol that randomizes the state sent to Bob. Let $\bP\in\{\pm\bX, \pm\bY, \pm\bZ\}$ be a Pauli-like observable, $\ket{\bP,\pm}$ its $\pm1$-eigenvectors, and consider the MUB-based Protocol~\ref{protocol:1w_mub}.
This protocol does not require entanglement and its equivalence to Protocol~\ref{protocol:1w_ent} can be easily shown, since each measurement on Alice's side in the latter determines an instance of the randomly chosen observable $\bP$ in the former  (see Appendix~\ref{app:equivalence} for a proof):
\begin{lemma}\label{lemma:equivalence}
    Protocols~\ref{protocol:1w_ent} and~\ref{protocol:1w_mub} are equivalent, up to a rescaling of the safety threshold $\epsilon^2 = 3\bar\epsilon^2/2$.
\end{lemma}

Finally, observe that both protocols can be carried out in a permutation-invariant way with respect to the different rounds, thanks to the randomization of measurements and actions that Alice and Bob undertake; this is crucial for a reduction of GC attacks to individual ones. 
\onecolumngrid
\begin{centering} 
\begin{minipage}{\textwidth}
\begin{protocol}[One-way MUB-based DQS]{protocol:1w_mub}
\item\label{step:preparation} \textbf{Preparation:}\\
Alice prepares the $n$-qubit probe state $\ket{\bP,+}$ for $\bP\in\{\pm\bX, \pm\bY, \pm\bZ\}$ chosen uniformly at random and sends it to Bob via the quantum channel. \\
\item \textbf{Encoding:}\label{step:encoding}\\
Bob chooses to perform one of these three actions at random:
\begin{enumerate}[(a)]
    \item with probability $p_c$ he does not encode the phase;\\
    \item with probability $p_e$, he encodes the phase by applying $U(\phi)$ to the received state;\\
    \item with probability $p_d = 1-p_c-p_e$, he discards the received state and moves on to the next round of the protocol.\\
\end{enumerate}
\item\label{step:measurement} \textbf{Measurement:}\\
Bob performs one of these two actions, depending on what he did at step~\ref{step:encoding}:
\begin{enumerate}[(a)]
    \item if he did not encode the phase, he measures one of the observables $\bX, \bY, \bZ$, chosen uniformly at random;\\
    \item if he encoded the phase, then he measures one of the observables $\bX, \bZ$, chosen uniformly at random;\\
\end{enumerate}
\item \textbf{Reconciliation and sifting:}\label{step:reconciliation} \\
After $T$ repetitions of steps \ref{step:preparation}-\ref{step:measurement}, Alice communicates to Bob the probe state $\ket{\bP,+}$ that she employed at each round. \\
Bob then keeps all rounds in which he measured $\pm\bP$, discarding the others.
\item \textbf{Fidelity check:}\\
From the rounds (a) where he did not encode the phase, Bob estimates the fidelity of the received state with the ideal state $\ket{\bP,+}$ as $\hat F_\bP = (\ave{\bP}+1)/2$, where the expectation is taken with respect to the tampered state. \\
If $\hat F_\bP \geq 1-\bar\epsilon^2$ for all $\bP$ and a given safety threshold $\bar\epsilon$, then Bob proceeds to the next step \ref{step:estimation}. \\
Otherwise, he aborts the protocol.
\item\label{step:estimation} \textbf{Phase estimation:}\\
From the rounds (b) where he applied $U(\phi)$, Bob estimates the phase as
\begin{equation}
    \arccos\left(\frac{1}{4}\sum_{\bP\in\{\pm\bX,\pm\bZ\}}\ave{\bP}\right),
\end{equation}
where the expectation is taken with respect to the phase-encoded tampered state.
\end{protocol}
\end{minipage}
\end{centering}
\twocolumngrid

\subsubsection{Quantifying tampering under GC attacks}
The protocol's faithfulness under GC attacks relies on a de-Finetti theorem for adaptive measurements. The basic idea is that, if we look at a small part of a large entangled system, this is approximately separable. At variance with the standard~\cite{Renner2007} for QKD, here we employ a more refined version of the theorem, due to ~\cite{Li2015} (see also~\cite{Brandao2013}), which applies to states processed using $\locc$ measurements, i.e., local measurements performed in a sequential adaptive fashion; these are defined, along with a suitable distance measures, as follows:

\begin{definition}
    ($\locc$ channel and distance)\\
    Let $\Lambda$ be a quantum channel (CPTP map) acting on $\D(\cH^{\otimes n})$. Then $\Lambda$ is $\locc$ if it can be written as 
     \begin{equation}
        \Lambda = \Lambda_1\otimes\cdots\otimes\Lambda_n,
    \end{equation}
   where $\Lambda_{j}$ for all $j=1,\cdots,n$, are measurement channels, i.e.,
    \begin{equation}
        \Lambda_j(\rho) = \sum_{i}\dketbra{i} \Tr[L_i^{(j)} \rho]
    \end{equation}
    for $\{L_i^{(j)}\}_{i}$ a POVM and $\{\ket i\}_i$ a set of orthonormal states, such that $\Lambda_j = \Lambda_{j|1,\cdots,j-1}$ can depend on the measurement outcomes on the previous $j-1$ parties. 

    Given two quantum states $\rho, \sigma\in\D(\cH^{\otimes n})$, their $\locc$-distance is
    \begin{equation}\label{eq:def_locc_distance}
       D_{\locc}(\rho,\sigma) = \max_{\Lambda\text{  }\locc}  D(\Lambda(\rho),\Lambda(\sigma)).
    \end{equation}
\end{definition}
Clearly, $D_{\locc}$ satisfies the standard properties of a distance. Moreover, $D_{\locc}\leq D$ and a suitable gentle-measurement lemma~\cite{winter1999} for $D_{\locc}$ holds (see Appendix~\ref{app:proofs}):
\begin{lemma}\label{lemma:gentle_measurement}
    (Gentle measurement for $\locc$-distance)\\
    Let $\1\geq E\geq0$ be a measurement operator and $\tau,\sigma$ two quantum states. Then
    \begin{equation}
        \Tr[E\sigma]\geq \Tr[E\tau]-2 D_{\locc}(\tau,\sigma).
    \end{equation}
\end{lemma}
Furthermore, note that our Protocol~\ref{protocol:1w_mub} is clearly contained in the $\locc$ measurement subset, which guarantees a better scaling of the error term. 

We will make use of the following de Finetti theorem for the $\locc$-distance:
\begin{theorem}\label{thm:de_finetti}
    (Quantum de Finetti for $\locc$-distance~\cite[Theorem 1]{Li2015})\\
    Let $\sigma\in\D(\cH^{\otimes n})$ be a permutation-invariant state, and $d={\rm dim}(\cH)$. Then for any integer $1\leq k\leq m$ there exists a measure $\mu$ on $\D(\cH)$ such that
    \begin{equation}
        D_{\locc}\left(\Tr_{m-k}[\sigma],\int d\mu(\rho) \rho^{\otimes k}\right) \leq (k-1)\sqrt{\frac{\log d}{2(m-k)}}.
    \end{equation}
\end{theorem}
Using these instruments, we prove that Bob or Alice can bound the fidelity to the ideal state on the $N_e$ estimation rounds via the the same quantity on the $N_c$ check rounds. This, in turn, implies a bound on the $\locc$-distance that quantifies Eve's tampering. 
\begin{theorem}\label{thm:distance_bound_1w_mub}(Tampering quantification for MUB-based protocol under GC attacks)\\
In Protocol~\ref{protocol:1w_mub}, let $N_c\geq N_e$ , $V(\phi) = U(\phi)^{\otimes N_e}$, and consider a random probe-state $\ket{\bP,+}$ prepared by Alice, for any $\bP\in\{\pm\bX,\pm\bY,\pm\bZ\}$. Then the $\locc$-distance between the ideal and tampered state after phase-encoding can be bounded in terms of the fidelity $\hat F_\bP$ as follows:
\begin{equation}
\begin{aligned}\label{eq:distance_bound_1w_mub_gc}
    & D_{\locc}(V(\phi)\ket{\bP,+}^{\otimes N_e},V(\phi)\Tr_{T-N_e}[\tilde\sigma] V(\phi)^\dagger)\\
    &\leq \sqrt{1-\hat F_\bP +4f(T,N_d,n)}.
\end{aligned}\end{equation}
where $\tilde\sigma = \Phi((\dketbra{\bP,+})^{\otimes T})$ is the collective tampered state after $T$ rounds of the protocol, with $\Phi$ the Eve-controlled collective quantum channel.
\end{theorem}
\begin{proof}
    Define $\ket\psi = \ket{\bP,+}$. We want to bound the $\locc$-distance between the ideal and tampered state during estimation rounds:
\begin{align}\label{eq:bound_with_fidelity_estimation}
    & D_{\locc}(V(\phi)\ket{\psi}^{\otimes N_e},V(\phi)\Tr_{T-N_e}[\tilde\sigma] V(\phi)^\dagger) \\
    &= D_{\locc}(\ket{\psi}^{\otimes N_e},\Tr_{T-N_e}[\tilde\sigma] ) \\
    &\leq \sqrt{1-F(\ket{\psi}^{\otimes N_e},\Tr_{T-N_e}[\tilde\sigma] ) }.
\end{align}
Since Bob's actions over $T$ rounds are permutation-invariant, the state $\tilde\sigma$ can be taken permutation-invariant too without loss of generality, hence Theorem~\ref{thm:de_finetti}  (with $m=T$, $k=T-N_d$ and $d=2^n$) implies the existence of a measure $d\mu(\rho)$ on states $\rho\in\D(\cH)$ such that{\small
\begin{equation}\label{eq:de_finetti_applied}
    D_{\locc}\left(\tilde\sigma,\int d\mu(\rho) \rho^{\otimes T}\right) 
   \leq f(T,N_d,n):=(T-N_d-1)\sqrt{\frac{n}{2 N_d}}.
\end{equation}
}We then apply Lemma \ref{lemma:gentle_measurement} twice: firstly, for the check rounds with $\tau=\Tr_{T-N_c}[\tilde\sigma]$, $\sigma=\int d\mu(\rho) \rho^{\otimes N_c}$ and $E=(\dketbra{\psi})^{\otimes N_c}$, obtaining
{\small
\begin{equation}\label{eq:gentle_check}
    \int d\mu(\rho) F(\ket{\psi},\rho)^{N_c} \geq F(\ket{\psi}^{\otimes N_c},\Tr_{T-N_c}[\tilde\sigma]) - 2 f(T,N_d,n),
\end{equation}
}where in the second term we have used \eqref{eq:de_finetti_applied} and the non-increasing property of trace distance under the map $\Tr_{T-N_c}$; secondly, for the estimation rounds with $\tau=\int d\mu(\rho) \rho^{\otimes N_e}$, $\sigma=\Tr_{T-N_e}[\tilde\sigma]$ and $E=(\dketbra{\psi})^{\otimes N_e}$, obtaining
 {\small
 \begin{align}\label{eq:gentle_estimation}
     F(\ket{\psi}^{\otimes N_e},\Tr_{T-N_e}[\tilde\sigma] ) \geq  \int d\mu(\rho) F(\ket{\psi},\rho)^{N_e} - 2 f(T,N_d,n). 
 \end{align}
}Finally, putting (\ref{eq:gentle_check},\ref{eq:gentle_estimation}) together and using the fact that $x^{N_e}\geq x^{N_c}$ for $1\geq x\geq 0$ and $N_c\geq N_e$ we obtain
{\small
\begin{equation}\label{eq:fidelity_1w_gc}
    F(\ket{\psi}^{\otimes N_e},\Tr_{T-N_e}[\tilde\sigma] ) \geq F(\ket{\psi}^{\otimes N_c},\Tr_{T-N_c}[\tilde\sigma]) - 4 f(T,N_d,n).
\end{equation}
}In the above expression, the first term on the right-hand side can be estimated by $\hat F_\bP$. Hence, plugging the expression into  \eqref{eq:bound_with_fidelity_estimation}, we obtain \eqref{eq:distance_bound_1w_mub_gc}.
\end{proof}
The bound ensures that the estimate of Eve's tampering during check rounds is also a good estimate for the estimation rounds, even in the case of GC attacks that can introduce correlations between different rounds. This can be guaranteed by making $f$ sufficiently small, i.e., discarding $N_d\gg N_c\geq N_e$ rounds. 

\subsubsection{$\locc$ bounds on bias and variance}
The last step is to employ Theorem~\ref{thm:distance_bound_1w_mub} to quantify how much the phase-estimation results using a tampered state differ from the ideal case, as measured by bias and variance to the true phase. In order to do this, we first have to show that the $\locc$-distance can be used to bound the deviation of local observables on collective states.
\begin{lemma}\label{lemma:unif_cont}
    (Uniform continuity of expectation with respect to $\locc$-distance)\\
    Let $A$ be a local observable on $\D(\cH^{\otimes m})$, given by the sum of at most $K$ product observables bounded by $a$, i.e., $A= \sum_{k=1}^K A_{1,k}\otimes\cdots\otimes A_{m,k}$ with $A_{j,k}$ observables on $\D(\cH)$ for all $j=1,\cdots,m$ and $k=1,\cdots,K$, and $||A_{1,k}\otimes\cdots\otimes A_{m,k}||_\infty\leq a$. Then for any two states $\sigma,\sigma'\in\D(\cH^{\otimes m})$ it holds
    \begin{equation}
        \left|\Tr[A(\sigma-\sigma')]\right| \leq 2 a K\cdot D_{\locc}(\sigma,\sigma').
    \end{equation}
\end{lemma}
\begin{proof}
    Let $A_{j,k}=\sum_{i=1}^d a_{i,j,k} \dketbra{a_{i,j,k}}$ be the spectral decomposition of the local observables constituting $A$,  and define the measurement channels $\Omega_{j,k}(\rho) = \sum_{i=1}^d \dketbra{i} \bra{a_{i,j,k}}\rho \ket{a_{i,j,k}}$ for all $\rho\in\D(\cH)$, $j=1,\cdots,m$ and $k=1,\cdots,K$.  Then for any $\sigma,\sigma'\in\D(\cH^{\otimes n})$ it holds
    \begin{align}
        \left|\Tr[A(\sigma-\sigma')]\right| & \leq \sum_{k=1}^K\left|\Tr[(A_{1,k}\otimes\cdots\otimes A_{m,k})(\sigma-\sigma')]\right|\\
         &\leq a\sum_{k=1}^K  ||(\Omega_{1,k}\otimes\cdots\otimes\Omega_{m,k})(\sigma-\sigma')||_1 \\
         &\leq 2aK\, \max_{\Omega\, \locc} D(\Omega(\sigma),\Omega(\sigma'))\\
         &= 2aK\, D_{\locc}(\sigma,\sigma').
    \end{align}
\end{proof}
In particular, if an estimation procedure is based on the measurement of a local observable $O$, which is the case of our protocols, we can use Lemma~\ref{lemma:unif_cont} to bound the difference in bias and variance of the estimation with respect to an ideal state $\sigma$, when the latter is tampered into $\sigma'$. This generalizes~\cite[Theorem 1]{Shettell2022a} to the case of non-iid samples which are close in $\locc$-distance (we defer the proof to the Appendix~\ref{app:proofs}). 
\begin{theorem}\label{thm:estimation_under_locc_tampering}
    (Estimation performance under $\locc$-tampering)\\
Let $O$ be an $o$-bounded observable on $\cH$, and $\sigma, \sigma'\in\D(\cH^{\otimes m})$ two quantum states such that $D_{\locc}(\sigma,\sigma')\leq \epsilon$. 
Define $\bar O = \frac1m\sum_{j=1}^{m}O_j\otimes\1^{\otimes (m-1)}$ as the $m$-partite average of $O$,  $\E_\sigma \bar O = \Tr[\bar O\sigma]$ as its expectation on $\sigma$, and $\Delta^2_{\sigma_0,\sigma} \bar O =\Tr[(\bar O-\E_{\sigma_0}\bar  O)^2 \sigma]$ as the variance of its estimation on $\sigma$ with respect to the state $\sigma_0$.

Then, the difference in the expectation of $\bar O$ on the two states is bounded as
\begin{equation}\label{eq:bias_bound_gc}
    \left|\E_{\sigma} \bar O-\E_{\sigma'} \bar O\right| \leq 2 o \epsilon,
\end{equation}
while the difference in the variance of the estimation of $\bar O$ with respect to the state $\sigma$ is bounded as
\begin{equation}
        \left|\Delta^2_{\sigma,\sigma} \bar O-\Delta^2_{\sigma,\sigma'} \bar O\right| \leq 4 o^2 (2\epsilon+\epsilon^2).
\end{equation}
\end{theorem}
Importantly, at variance with~\cite{Shettell2022a}, here the bound on the variance has a dependence on $\epsilon$ rather than $\epsilon/m$. This is due to our relaxation of the iid assumption for the states $\sigma, \sigma'$, which is necessary to contemplate GC attacks.\\

\subsubsection{Security and faithfulness for MUB- and entanglement-based protocols}
We are now ready to prove the security and faithfulness of the MUB-based protocol:
\begin{theorem}
    \label{thm:main_mub}
    The one-way MUB-based DQS Protocol~\ref{protocol:1w_mub} is perfectly secure and its faithfulness is determined by $\epsilon$. Indeed, the difference of bias and variance of the unbiased estimator $\hat\phi$ calculated on the ideal probe state with respect to the same estimator $\hat\phi'$ on the tampered state are bounded, under general-coherent attacks, as
\begin{align}
    &\left|\E \hat\phi-\E\hat \phi'\right| \leq \frac{\epsilon_0}{n |\sin(2n\phi)|}, \label{eq:bias_gc_mub}\\
    & \left|\Delta^2\hat\phi-\Delta^2 \hat\phi'\right| \leq \frac{(2\epsilon_0+\epsilon_0^2)}{n^2 \sin^2(2n\phi)}.\label{eq:var_gc_mub}
\end{align}
with $\epsilon_0 = \sqrt{\bar\epsilon^2 + 4 f(T,N_d,n)}$.
In the case of individual attacks, one can take $N_d=0$ and employ the same bias bound of  \eqref{eq:bias_gc}  with $\epsilon_0\mapsto\bar\epsilon$, while the variance bound reads
\begin{equation}
    \left|\Delta^2\hat\phi-\Delta^2 \hat\phi'\right| \leq \frac{ (2\bar\epsilon/N_e+\bar\epsilon^2)}{n^2 \sin^2(2n\phi)}.
\end{equation}
The two-way version of the protocol has similar bounds with $\epsilon_0\mapsto \sqrt{\bar\epsilon^2 + 2 f(T,N_d,n)} + |\sin(\min\left\{n\phi,\pi/2\right\})|$, where $\delta(U)$ is the minimum arc-length on the unit-circle that contains all eigenvalues of $U$.
\end{theorem}
\begin{proof}
Perfect security in the one-way protocol is straightforward to prove. Indeed, even though Eve can fully control the quantum state received by Bob, e.g., by replacing it with a known probe state, she cannot measure said state nor access Bob's measurement outcomes to estimate the phase. 

As for faithfulness, if the check step passes, the protocol guarantees that $\hat F_\bP\geq 1-\bar\epsilon^2$ for all $\bP$. Hence, by Theorem~\ref{thm:distance_bound_1w_mub}, the estimation rounds are $\epsilon_0$-close to ideal for all $\bP$, i.e.,
\begin{equation}
    D_{\locc}(V(\phi)\ket{\bP,+}^{\otimes N_e},V(\phi)\Tr_{T-N_e}[\tilde\sigma] V(\phi)^\dagger)\leq \epsilon_0.
\end{equation}
Combining this estimate with Theorem~\ref{thm:estimation_under_locc_tampering}, with 
$O = \bP\in\{\pm\bX,\pm\bZ\}$ and $o=1$ we obtain
\begin{align}\label{eq:bias_bound_gc}
    &\left|\E_{\sigma} \bar O-\E_{\sigma'} \bar O\right| \leq 2 \epsilon_0,\\
      &  \left|\Delta^2_{\sigma,\sigma} \bar O-\Delta^2_{\sigma,\sigma'} \bar O\right| \leq 4 (2\epsilon_0+\epsilon_0^2).
\end{align}
The bias and variance on the estimator are obtained via error propagation, rescaling by
\begin{equation}
   | \partial_\theta \ave{O}| = |\partial_\theta \cos(n\phi)| = n |\sin(n\phi)|,
\end{equation}
or its square, where the expectation is taken with respect to the phase-encoded ideal state $e^{-\ii\frac\phi2 \bY}\ket{\bP,+}$ for each $\bP$. 

If we restrict to individual attacks, Theorem~\ref{thm:distance_bound_1w_mub}
simplifies considerably since the collective state $\tilde\sigma = \tilde\rho^{\otimes T}$ is a product state with respect to each round, so that we can apply directly the bias and variance bounds of~\cite{Shettell2022a}. Then, calling $W(\phi) = e^{-\ii \frac\phi2 Y}$, we have
\begin{align}\label{eq:bound_with_fidelity_estimation}
    & D(W(\phi)\ket{\bP,+},W(\phi)\tilde\rho W(\phi)^\dagger) \\
    &\leq \sqrt{1-F(\ket{\bP,+},\tilde\rho ) } \leq \bar\epsilon,
\end{align}
where we have used the unitary-invariance of the trace-distance and the Fuchs-Van-der-Graf inequality, and the fidelity to the ideal state is estimated as $\hat F_\bP$ from the check rounds. 

The proof of the two-way case is deferred to the Appendix~\ref{app:2w}.
 \end{proof}
The proof of Theorem~\ref{thm:main} follows straightforwardly by combining Theorem~\ref{thm:main_mub} with Lemma~\ref{lemma:equivalence} above. 

\subsection{Experimental methods}

The photon source is a 3-mm beta-barium borate crystal pumped by a CW laser at 405 nm - see Fig.~\ref{fig:exp_setup}. The downconversion emission is collected by single-mode fibres through interference filters with 7.5 nm full width at half maximum. Photons are then delivered to a partially polarising beam splitter (PPBS) $T_H=1,T_V=1/3$), both with the polarisation $\ket{\tilde D} = (\ket{H}+\sqrt{3}\ket{V})/2$ ($H$ is the horizontal polarisation, $V$ the vertical, $D$ diagonal at 45$^\circ$, $A$ diagonal at 135$^\circ$, $L$ left-circular, $R$ right-circular). This is chosen in such a way to produce a maximally entangled state $\ket{\psi}=\left(\ket{HD}+{VA}\right)/\sqrt{2}$ by nonclassical interference, postselecting on coincidence events~\cite{Langford05,Kiesel05,Okamoto05}. The system is embedded in a Sagnac interferometer in order to curtail the impact of imperfect transimittivities of the PPBS (not shown in the figure) \cite{bizzarri2024steering,bizzarri2024quasiprob,bizzarri2025sloppy}. 

Polarisation is then analysed by means of a standard arrangement of a half wave-plate (HWP), a quarter wave-plate (QWP), and a polarising beam splitter (PBS). Acquisition windows are set as 0.1 ns, each containing around 100 coincidence detection events. Each correlator, used to either the check or the estimation rounds, is reconstructed as  
\begin{align}
    \ave{A\otimes B} = \frac{\sum_{i,j}(-1)^{i+j}N_{ij}}{\sum_{i,j}N_{ij}}\quad i,j = 0,1
    \label{eq:correlators_computation}
\end{align}
where $A$ and $B$ are the proper Pauli operators $X,Y$ or $Z$, and $N_{i,j}$ are the coincidence counts registered - $i=0$ corresponds to $D,L$ or $H$, $i=1$ to $A,R$ or $V$, respectively for the three observables. Fidelity is estimated from events collected in 200 detection windows. The phases in Fig.~\ref{fig:dati} are the average over 200 values, each one derived from a distinct detection window. The uncertainty on the phase is obtained by error propagation on \eqref{eq:correlators_computation}, leading to 
\begin{align}
 \Delta^2\hat\phi' = \frac{\Delta^2(Z\otimes X)+\Delta^2(X\otimes Z)}{4\left( 4-\left(\overline{\ave{X\otimes Z}+\ave{Z\otimes X}}\right)^2\right)}.
\end{align}
Here the bar stands for the average over the 200 experimental values.

\section*{Data availability statement} 
The data supporting our findings are available on a public-access repository at~\cite{Repo}.

\section*{Acknowledgements} 
We thank F. De Stefani for valuable discussion, and A. Fabbri for assistance with malfunctioning equipment.
This work was supported by the PRIN project PRIN22-RISQUE-2022T25TR3 of the Italian Ministry of University. G.B. is supported by Rome Technopole Innovation Ecosystem (PNRR grant M4-C2-Inv). IG acknowledges the support from MUR Dipartimento di Eccellenza 2023-2027. M.R. acknowledges support from the project PNRR - Finanziato dall’Unione Europea - MISSIONE 4 COMPONENTE 2 INVESTIMENTO 1.2 - “Finanziamento di progetti presentati da giovani ricercatori” - Id MSCA 0000011-SQUID - CUP F83C22002390007 (Young Researchers) - Finanziato dall’Unione europea - NextGenerationEU. 

\bibliography{library.bib,library2.bib}

\begin{thebibliography}{44}%
\makeatletter
\providecommand \@ifxundefined [1]{%
 \@ifx{#1\undefined}
}%
\providecommand \@ifnum [1]{%
 \ifnum #1\expandafter \@firstoftwo
 \else \expandafter \@secondoftwo
 \fi
}%
\providecommand \@ifx [1]{%
 \ifx #1\expandafter \@firstoftwo
 \else \expandafter \@secondoftwo
 \fi
}%
\providecommand \natexlab [1]{#1}%
\providecommand \enquote  [1]{``#1''}%
\providecommand \bibnamefont  [1]{#1}%
\providecommand \bibfnamefont [1]{#1}%
\providecommand \citenamefont [1]{#1}%
\providecommand \href@noop [0]{\@secondoftwo}%
\providecommand \href [0]{\begingroup \@sanitize@url \@href}%
\providecommand \@href[1]{\@@startlink{#1}\@@href}%
\providecommand \@@href[1]{\endgroup#1\@@endlink}%
\providecommand \@sanitize@url [0]{\catcode `\\12\catcode `\$12\catcode `\&12\catcode `\#12\catcode `\^12\catcode `\_12\catcode `\%12\relax}%
\providecommand \@@startlink[1]{}%
\providecommand \@@endlink[0]{}%
\providecommand \url  [0]{\begingroup\@sanitize@url \@url }%
\providecommand \@url [1]{\endgroup\@href {#1}{\urlprefix }}%
\providecommand \urlprefix  [0]{URL }%
\providecommand \Eprint [0]{\href }%
\providecommand \doibase [0]{https://doi.org/}%
\providecommand \selectlanguage [0]{\@gobble}%
\providecommand \bibinfo  [0]{\@secondoftwo}%
\providecommand \bibfield  [0]{\@secondoftwo}%
\providecommand \translation [1]{[#1]}%
\providecommand \BibitemOpen [0]{}%
\providecommand \bibitemStop [0]{}%
\providecommand \bibitemNoStop [0]{.\EOS\space}%
\providecommand \EOS [0]{\spacefactor3000\relax}%
\providecommand \BibitemShut  [1]{\csname bibitem#1\endcsname}%
\let\auto@bib@innerbib\@empty
\bibitem [{\citenamefont {Giovannetti}\ \emph {et~al.}(2011)\citenamefont {Giovannetti}, \citenamefont {Lloyd},\ and\ \citenamefont {Maccone}}]{Giovannetti2011}%
  \BibitemOpen
  \bibfield  {author} {\bibinfo {author} {\bibfnamefont {V.}~\bibnamefont {Giovannetti}}, \bibinfo {author} {\bibfnamefont {S.}~\bibnamefont {Lloyd}},\ and\ \bibinfo {author} {\bibfnamefont {L.}~\bibnamefont {Maccone}},\ }\bibfield  {title} {\bibinfo {title} {{Advances in quantum metrology}},\ }\href {https://doi.org/10.1038/nphoton.2011.35} {\bibfield  {journal} {\bibinfo  {journal} {Nat. Photonics}\ }\textbf {\bibinfo {volume} {5}},\ \bibinfo {pages} {222} (\bibinfo {year} {2011})}\BibitemShut {NoStop}%
\bibitem [{\citenamefont {Yin}\ \emph {et~al.}(2020)\citenamefont {Yin}, \citenamefont {Takeuchi}, \citenamefont {Zhang}, \citenamefont {Yin}, \citenamefont {Matsuzaki}, \citenamefont {Peng}, \citenamefont {Xu}, \citenamefont {Xu}, \citenamefont {Tang}, \citenamefont {Zhou}, \citenamefont {Chen}, \citenamefont {Li},\ and\ \citenamefont {Guo}}]{Yin2020}%
  \BibitemOpen
  \bibfield  {author} {\bibinfo {author} {\bibfnamefont {P.}~\bibnamefont {Yin}}, \bibinfo {author} {\bibfnamefont {Y.}~\bibnamefont {Takeuchi}}, \bibinfo {author} {\bibfnamefont {W.-H.}\ \bibnamefont {Zhang}}, \bibinfo {author} {\bibfnamefont {Z.-Q.}\ \bibnamefont {Yin}}, \bibinfo {author} {\bibfnamefont {Y.}~\bibnamefont {Matsuzaki}}, \bibinfo {author} {\bibfnamefont {X.-X.}\ \bibnamefont {Peng}}, \bibinfo {author} {\bibfnamefont {X.-Y.}\ \bibnamefont {Xu}}, \bibinfo {author} {\bibfnamefont {J.-S.}\ \bibnamefont {Xu}}, \bibinfo {author} {\bibfnamefont {J.-S.}\ \bibnamefont {Tang}}, \bibinfo {author} {\bibfnamefont {Z.-Q.}\ \bibnamefont {Zhou}}, \bibinfo {author} {\bibfnamefont {G.}~\bibnamefont {Chen}}, \bibinfo {author} {\bibfnamefont {C.-F.}\ \bibnamefont {Li}},\ and\ \bibinfo {author} {\bibfnamefont {G.-C.}\ \bibnamefont {Guo}},\ }\bibfield  {title} {\bibinfo {title} {{Experimental Demonstration of Secure Quantum Remote Sensing}},\ }\href {https://doi.org/10.1103/PhysRevApplied.14.014065} {\bibfield
  {journal} {\bibinfo  {journal} {Phys. Rev. Appl.}\ }\textbf {\bibinfo {volume} {14}},\ \bibinfo {pages} {014065} (\bibinfo {year} {2020})}\BibitemShut {NoStop}%
\bibitem [{\citenamefont {Takeuchi}\ \emph {et~al.}(2019)\citenamefont {Takeuchi}, \citenamefont {Matsuzaki}, \citenamefont {Miyanishi}, \citenamefont {Sugiyama},\ and\ \citenamefont {Munro}}]{Takeuchi2019}%
  \BibitemOpen
  \bibfield  {author} {\bibinfo {author} {\bibfnamefont {Y.}~\bibnamefont {Takeuchi}}, \bibinfo {author} {\bibfnamefont {Y.}~\bibnamefont {Matsuzaki}}, \bibinfo {author} {\bibfnamefont {K.}~\bibnamefont {Miyanishi}}, \bibinfo {author} {\bibfnamefont {T.}~\bibnamefont {Sugiyama}},\ and\ \bibinfo {author} {\bibfnamefont {W.~J.}\ \bibnamefont {Munro}},\ }\bibfield  {title} {\bibinfo {title} {{Quantum remote sensing with asymmetric information gain}},\ }\href {https://doi.org/10.1103/PhysRevA.99.022325} {\bibfield  {journal} {\bibinfo  {journal} {Phys. Rev. A}\ }\textbf {\bibinfo {volume} {99}},\ \bibinfo {pages} {022325} (\bibinfo {year} {2019})}\BibitemShut {NoStop}%
\bibitem [{\citenamefont {Moore}\ and\ \citenamefont {Dunningham}(2023)}]{Moore2023}%
  \BibitemOpen
  \bibfield  {author} {\bibinfo {author} {\bibfnamefont {S.~W.}\ \bibnamefont {Moore}}\ and\ \bibinfo {author} {\bibfnamefont {J.~A.}\ \bibnamefont {Dunningham}},\ }\bibfield  {title} {\bibinfo {title} {{Secure quantum remote sensing without entanglement}},\ }\bibfield  {journal} {\bibinfo  {journal} {AVS Quantum Sci.}\ }\textbf {\bibinfo {volume} {5}},\ \href {https://doi.org/10.1116/5.0137260} {10.1116/5.0137260} (\bibinfo {year} {2023}),\ \Eprint {https://arxiv.org/abs/2302.03617} {arXiv:2302.03617} \BibitemShut {NoStop}%
\bibitem [{\citenamefont {Huang}\ \emph {et~al.}(2019)\citenamefont {Huang}, \citenamefont {Macchiavello},\ and\ \citenamefont {Maccone}}]{Huang2019a}%
  \BibitemOpen
  \bibfield  {author} {\bibinfo {author} {\bibfnamefont {Z.}~\bibnamefont {Huang}}, \bibinfo {author} {\bibfnamefont {C.}~\bibnamefont {Macchiavello}},\ and\ \bibinfo {author} {\bibfnamefont {L.}~\bibnamefont {Maccone}},\ }\bibfield  {title} {\bibinfo {title} {{Cryptographic quantum metrology}},\ }\href {https://doi.org/10.1103/PhysRevA.99.022314} {\bibfield  {journal} {\bibinfo  {journal} {Phys. Rev. A}\ }\textbf {\bibinfo {volume} {99}},\ \bibinfo {pages} {022314} (\bibinfo {year} {2019})}\BibitemShut {NoStop}%
\bibitem [{\citenamefont {Shettell}\ \emph {et~al.}(2022{\natexlab{a}})\citenamefont {Shettell}, \citenamefont {Kashefi},\ and\ \citenamefont {Markham}}]{Shettell2022a}%
  \BibitemOpen
  \bibfield  {author} {\bibinfo {author} {\bibfnamefont {N.}~\bibnamefont {Shettell}}, \bibinfo {author} {\bibfnamefont {E.}~\bibnamefont {Kashefi}},\ and\ \bibinfo {author} {\bibfnamefont {D.}~\bibnamefont {Markham}},\ }\bibfield  {title} {\bibinfo {title} {{Cryptographic approach to quantum metrology}},\ }\href {https://doi.org/10.1103/PhysRevA.105.L010401} {\bibfield  {journal} {\bibinfo  {journal} {Phys. Rev. A}\ }\textbf {\bibinfo {volume} {105}},\ \bibinfo {pages} {L010401} (\bibinfo {year} {2022}{\natexlab{a}})}\BibitemShut {NoStop}%
\bibitem [{\citenamefont {Rogers}(1999)}]{Alan_Rogers_1999}%
  \BibitemOpen
  \bibfield  {author} {\bibinfo {author} {\bibfnamefont {A.}~\bibnamefont {Rogers}},\ }\bibfield  {title} {\bibinfo {title} {Distributed optical-fibre sensing},\ }\href {https://doi.org/10.1088/0957-0233/10/8/201} {\bibfield  {journal} {\bibinfo  {journal} {Measurement Science and Technology}\ }\textbf {\bibinfo {volume} {10}},\ \bibinfo {pages} {R75} (\bibinfo {year} {1999})}\BibitemShut {NoStop}%
\bibitem [{\citenamefont {Qi}\ \emph {et~al.}(2001)\citenamefont {Qi}, \citenamefont {Iyengar},\ and\ \citenamefont {Chakrabarty}}]{QI2001655}%
  \BibitemOpen
  \bibfield  {author} {\bibinfo {author} {\bibfnamefont {H.}~\bibnamefont {Qi}}, \bibinfo {author} {\bibfnamefont {S.}~\bibnamefont {Iyengar}},\ and\ \bibinfo {author} {\bibfnamefont {K.}~\bibnamefont {Chakrabarty}},\ }\bibfield  {title} {\bibinfo {title} {Distributed sensor networks—a review of recent research},\ }\href {https://doi.org/https://doi.org/10.1016/S0016-0032(01)00026-6} {\bibfield  {journal} {\bibinfo  {journal} {Journal of the Franklin Institute}\ }\textbf {\bibinfo {volume} {338}},\ \bibinfo {pages} {655} (\bibinfo {year} {2001})},\ \bibinfo {note} {distributed Sensor Networks for Real-time Systems with Adaptive C onfiguration}\BibitemShut {NoStop}%
\bibitem [{\citenamefont {He}\ and\ \citenamefont {Liu}(2021)}]{9354971}%
  \BibitemOpen
  \bibfield  {author} {\bibinfo {author} {\bibfnamefont {Z.}~\bibnamefont {He}}\ and\ \bibinfo {author} {\bibfnamefont {Q.}~\bibnamefont {Liu}},\ }\bibfield  {title} {\bibinfo {title} {Optical fiber distributed acoustic sensors: A review},\ }\href {https://doi.org/10.1109/JLT.2021.3059771} {\bibfield  {journal} {\bibinfo  {journal} {Journal of Lightwave Technology}\ }\textbf {\bibinfo {volume} {39}},\ \bibinfo {pages} {3671} (\bibinfo {year} {2021})}\BibitemShut {NoStop}%
\bibitem [{\citenamefont {Ukil}\ \emph {et~al.}(2012)\citenamefont {Ukil}, \citenamefont {Braendle},\ and\ \citenamefont {Krippner}}]{5955066}%
  \BibitemOpen
  \bibfield  {author} {\bibinfo {author} {\bibfnamefont {A.}~\bibnamefont {Ukil}}, \bibinfo {author} {\bibfnamefont {H.}~\bibnamefont {Braendle}},\ and\ \bibinfo {author} {\bibfnamefont {P.}~\bibnamefont {Krippner}},\ }\bibfield  {title} {\bibinfo {title} {Distributed temperature sensing: Review of technology and applications},\ }\href {https://doi.org/10.1109/JSEN.2011.2162060} {\bibfield  {journal} {\bibinfo  {journal} {IEEE Sensors Journal}\ }\textbf {\bibinfo {volume} {12}},\ \bibinfo {pages} {885} (\bibinfo {year} {2012})}\BibitemShut {NoStop}%
\bibitem [{\citenamefont {Proctor}\ \emph {et~al.}(2018)\citenamefont {Proctor}, \citenamefont {Knott},\ and\ \citenamefont {Dunningham}}]{PhysRevLett.120.080501}%
  \BibitemOpen
  \bibfield  {author} {\bibinfo {author} {\bibfnamefont {T.~J.}\ \bibnamefont {Proctor}}, \bibinfo {author} {\bibfnamefont {P.~A.}\ \bibnamefont {Knott}},\ and\ \bibinfo {author} {\bibfnamefont {J.~A.}\ \bibnamefont {Dunningham}},\ }\bibfield  {title} {\bibinfo {title} {Multiparameter estimation in networked quantum sensors},\ }\href {https://doi.org/10.1103/PhysRevLett.120.080501} {\bibfield  {journal} {\bibinfo  {journal} {Phys. Rev. Lett.}\ }\textbf {\bibinfo {volume} {120}},\ \bibinfo {pages} {080501} (\bibinfo {year} {2018})}\BibitemShut {NoStop}%
\bibitem [{\citenamefont {Ge}\ \emph {et~al.}(2018)\citenamefont {Ge}, \citenamefont {Jacobs}, \citenamefont {Eldredge}, \citenamefont {Gorshkov},\ and\ \citenamefont {Foss-Feig}}]{PhysRevLett.121.043604}%
  \BibitemOpen
  \bibfield  {author} {\bibinfo {author} {\bibfnamefont {W.}~\bibnamefont {Ge}}, \bibinfo {author} {\bibfnamefont {K.}~\bibnamefont {Jacobs}}, \bibinfo {author} {\bibfnamefont {Z.}~\bibnamefont {Eldredge}}, \bibinfo {author} {\bibfnamefont {A.~V.}\ \bibnamefont {Gorshkov}},\ and\ \bibinfo {author} {\bibfnamefont {M.}~\bibnamefont {Foss-Feig}},\ }\bibfield  {title} {\bibinfo {title} {Distributed quantum metrology with linear networks and separable inputs},\ }\href {https://doi.org/10.1103/PhysRevLett.121.043604} {\bibfield  {journal} {\bibinfo  {journal} {Phys. Rev. Lett.}\ }\textbf {\bibinfo {volume} {121}},\ \bibinfo {pages} {043604} (\bibinfo {year} {2018})}\BibitemShut {NoStop}%
\bibitem [{\citenamefont {Zhuang}\ \emph {et~al.}(2018)\citenamefont {Zhuang}, \citenamefont {Zhang},\ and\ \citenamefont {Shapiro}}]{PhysRevA.97.032329}%
  \BibitemOpen
  \bibfield  {author} {\bibinfo {author} {\bibfnamefont {Q.}~\bibnamefont {Zhuang}}, \bibinfo {author} {\bibfnamefont {Z.}~\bibnamefont {Zhang}},\ and\ \bibinfo {author} {\bibfnamefont {J.~H.}\ \bibnamefont {Shapiro}},\ }\bibfield  {title} {\bibinfo {title} {Distributed quantum sensing using continuous-variable multipartite entanglement},\ }\href {https://doi.org/10.1103/PhysRevA.97.032329} {\bibfield  {journal} {\bibinfo  {journal} {Phys. Rev. A}\ }\textbf {\bibinfo {volume} {97}},\ \bibinfo {pages} {032329} (\bibinfo {year} {2018})}\BibitemShut {NoStop}%
\bibitem [{\citenamefont {Kasai}\ \emph {et~al.}(2022)\citenamefont {Kasai}, \citenamefont {Takeuchi}, \citenamefont {Hakoshima}, \citenamefont {Matsuzaki},\ and\ \citenamefont {Tokura}}]{Kasai2022}%
  \BibitemOpen
  \bibfield  {author} {\bibinfo {author} {\bibfnamefont {H.}~\bibnamefont {Kasai}}, \bibinfo {author} {\bibfnamefont {Y.}~\bibnamefont {Takeuchi}}, \bibinfo {author} {\bibfnamefont {H.}~\bibnamefont {Hakoshima}}, \bibinfo {author} {\bibfnamefont {Y.}~\bibnamefont {Matsuzaki}},\ and\ \bibinfo {author} {\bibfnamefont {Y.}~\bibnamefont {Tokura}},\ }\bibfield  {title} {\bibinfo {title} {{Anonymous Quantum Sensing}},\ }\bibfield  {journal} {\bibinfo  {journal} {J. Phys. Soc. Japan}\ }\textbf {\bibinfo {volume} {91}},\ \href {https://doi.org/10.7566/JPSJ.91.074005} {10.7566/JPSJ.91.074005} (\bibinfo {year} {2022})\BibitemShut {NoStop}%
\bibitem [{\citenamefont {Guo}\ \emph {et~al.}(2020)\citenamefont {Guo}, \citenamefont {Breum}, \citenamefont {Borregaard}, \citenamefont {Izumi}, \citenamefont {Larsen}, \citenamefont {Gehring}, \citenamefont {Christandl}, \citenamefont {Neergaard-Nielsen},\ and\ \citenamefont {Andersen}}]{expDQS}%
  \BibitemOpen
  \bibfield  {author} {\bibinfo {author} {\bibfnamefont {X.}~\bibnamefont {Guo}}, \bibinfo {author} {\bibfnamefont {C.~R.}\ \bibnamefont {Breum}}, \bibinfo {author} {\bibfnamefont {J.}~\bibnamefont {Borregaard}}, \bibinfo {author} {\bibfnamefont {S.}~\bibnamefont {Izumi}}, \bibinfo {author} {\bibfnamefont {M.~V.}\ \bibnamefont {Larsen}}, \bibinfo {author} {\bibfnamefont {T.}~\bibnamefont {Gehring}}, \bibinfo {author} {\bibfnamefont {M.}~\bibnamefont {Christandl}}, \bibinfo {author} {\bibfnamefont {J.~S.}\ \bibnamefont {Neergaard-Nielsen}},\ and\ \bibinfo {author} {\bibfnamefont {U.~L.}\ \bibnamefont {Andersen}},\ }\bibfield  {title} {\bibinfo {title} {Distributed quantum sensing in a continuous-variable entangled network},\ }\href {https://doi.org/10.1038/s41567-019-0743-x} {\bibfield  {journal} {\bibinfo  {journal} {Nature Physics}\ }\textbf {\bibinfo {volume} {16}},\ \bibinfo {pages} {281} (\bibinfo {year} {2020})}\BibitemShut {NoStop}%
\bibitem [{\citenamefont {Zhao}\ \emph {et~al.}(2021)\citenamefont {Zhao}, \citenamefont {Zhang}, \citenamefont {Liu}, \citenamefont {Guan}, \citenamefont {Zhang}, \citenamefont {Li}, \citenamefont {Bai}, \citenamefont {Li}, \citenamefont {Liu}, \citenamefont {You}, \citenamefont {Zhang}, \citenamefont {Fan}, \citenamefont {Xu}, \citenamefont {Zhang},\ and\ \citenamefont {Pan}}]{Zhao2021a}%
  \BibitemOpen
  \bibfield  {author} {\bibinfo {author} {\bibfnamefont {S.-R.}\ \bibnamefont {Zhao}}, \bibinfo {author} {\bibfnamefont {Y.-Z.}\ \bibnamefont {Zhang}}, \bibinfo {author} {\bibfnamefont {W.-Z.}\ \bibnamefont {Liu}}, \bibinfo {author} {\bibfnamefont {J.-Y.}\ \bibnamefont {Guan}}, \bibinfo {author} {\bibfnamefont {W.}~\bibnamefont {Zhang}}, \bibinfo {author} {\bibfnamefont {C.-L.}\ \bibnamefont {Li}}, \bibinfo {author} {\bibfnamefont {B.}~\bibnamefont {Bai}}, \bibinfo {author} {\bibfnamefont {M.-H.}\ \bibnamefont {Li}}, \bibinfo {author} {\bibfnamefont {Y.}~\bibnamefont {Liu}}, \bibinfo {author} {\bibfnamefont {L.}~\bibnamefont {You}}, \bibinfo {author} {\bibfnamefont {J.}~\bibnamefont {Zhang}}, \bibinfo {author} {\bibfnamefont {J.}~\bibnamefont {Fan}}, \bibinfo {author} {\bibfnamefont {F.}~\bibnamefont {Xu}}, \bibinfo {author} {\bibfnamefont {Q.}~\bibnamefont {Zhang}},\ and\ \bibinfo {author} {\bibfnamefont {J.-W.}\ \bibnamefont {Pan}},\ }\bibfield  {title} {\bibinfo {title} {{Field Demonstration of Distributed
  Quantum Sensing without Post-Selection}},\ }\href {https://doi.org/10.1103/PhysRevX.11.031009} {\bibfield  {journal} {\bibinfo  {journal} {Phys. Rev. X}\ }\textbf {\bibinfo {volume} {11}},\ \bibinfo {pages} {031009} (\bibinfo {year} {2021})}\BibitemShut {NoStop}%
\bibitem [{\citenamefont {K{\'o}m{\'a}r}\ \emph {et~al.}(2014)\citenamefont {K{\'o}m{\'a}r}, \citenamefont {Kessler}, \citenamefont {Bishof}, \citenamefont {Jiang}, \citenamefont {S{\o}rensen}, \citenamefont {Ye},\ and\ \citenamefont {Lukin}}]{DQSclocks}%
  \BibitemOpen
  \bibfield  {author} {\bibinfo {author} {\bibfnamefont {P.}~\bibnamefont {K{\'o}m{\'a}r}}, \bibinfo {author} {\bibfnamefont {E.~M.}\ \bibnamefont {Kessler}}, \bibinfo {author} {\bibfnamefont {M.}~\bibnamefont {Bishof}}, \bibinfo {author} {\bibfnamefont {L.}~\bibnamefont {Jiang}}, \bibinfo {author} {\bibfnamefont {A.~S.}\ \bibnamefont {S{\o}rensen}}, \bibinfo {author} {\bibfnamefont {J.}~\bibnamefont {Ye}},\ and\ \bibinfo {author} {\bibfnamefont {M.~D.}\ \bibnamefont {Lukin}},\ }\bibfield  {title} {\bibinfo {title} {A quantum network of clocks},\ }\href {https://doi.org/10.1038/nphys3000} {\bibfield  {journal} {\bibinfo  {journal} {Nature Physics}\ }\textbf {\bibinfo {volume} {10}},\ \bibinfo {pages} {582} (\bibinfo {year} {2014})}\BibitemShut {NoStop}%
\bibitem [{\citenamefont {He}\ \emph {et~al.}(2024)\citenamefont {He}, \citenamefont {Huang}, \citenamefont {Guan}, \citenamefont {Chen}, \citenamefont {Zhang},\ and\ \citenamefont {Wei}}]{He2024}%
  \BibitemOpen
  \bibfield  {author} {\bibinfo {author} {\bibfnamefont {W.}~\bibnamefont {He}}, \bibinfo {author} {\bibfnamefont {C.}~\bibnamefont {Huang}}, \bibinfo {author} {\bibfnamefont {R.}~\bibnamefont {Guan}}, \bibinfo {author} {\bibfnamefont {Y.}~\bibnamefont {Chen}}, \bibinfo {author} {\bibfnamefont {Z.}~\bibnamefont {Zhang}},\ and\ \bibinfo {author} {\bibfnamefont {K.}~\bibnamefont {Wei}},\ }\href {http://arxiv.org/abs/2412.18837} {\bibinfo {title} {{Experimental secure entanglement-free quantum remote sensing over 50 km of optical fiber}}} (\bibinfo {year} {2024}),\ \Eprint {https://arxiv.org/abs/2412.18837} {arXiv:2412.18837} \BibitemShut {NoStop}%
\bibitem [{\citenamefont {Xie}\ \emph {et~al.}(2018)\citenamefont {Xie}, \citenamefont {Xu}, \citenamefont {Chen},\ and\ \citenamefont {Wang}}]{Xie18}%
  \BibitemOpen
  \bibfield  {author} {\bibinfo {author} {\bibfnamefont {D.}~\bibnamefont {Xie}}, \bibinfo {author} {\bibfnamefont {C.}~\bibnamefont {Xu}}, \bibinfo {author} {\bibfnamefont {J.}~\bibnamefont {Chen}},\ and\ \bibinfo {author} {\bibfnamefont {A.~M.}\ \bibnamefont {Wang}},\ }\bibfield  {title} {\bibinfo {title} {High-dimensional cryptographic quantum parameter estimation},\ }\href {https://doi.org/10.1007/s11128-018-1884-z} {\bibfield  {journal} {\bibinfo  {journal} {Quantum Information Processing}\ }\textbf {\bibinfo {volume} {17}},\ \bibinfo {pages} {116} (\bibinfo {year} {2018})}\BibitemShut {NoStop}%
\bibitem [{\citenamefont {Ho}\ \emph {et~al.}(2024)\citenamefont {Ho}, \citenamefont {Webb}, \citenamefont {Brooks}, \citenamefont {Grasselli}, \citenamefont {Gauger},\ and\ \citenamefont {Fedrizzi}}]{Ho2024}%
  \BibitemOpen
  \bibfield  {author} {\bibinfo {author} {\bibfnamefont {J.}~\bibnamefont {Ho}}, \bibinfo {author} {\bibfnamefont {J.~W.}\ \bibnamefont {Webb}}, \bibinfo {author} {\bibfnamefont {R.~M.~J.}\ \bibnamefont {Brooks}}, \bibinfo {author} {\bibfnamefont {F.}~\bibnamefont {Grasselli}}, \bibinfo {author} {\bibfnamefont {E.}~\bibnamefont {Gauger}},\ and\ \bibinfo {author} {\bibfnamefont {A.}~\bibnamefont {Fedrizzi}},\ }\href {http://arxiv.org/abs/2410.00970} {\bibinfo {title} {{Quantum-private distributed sensing}}} (\bibinfo {year} {2024}),\ \Eprint {https://arxiv.org/abs/2410.00970} {arXiv:2410.00970} \BibitemShut {NoStop}%
\bibitem [{\citenamefont {Hassani}\ \emph {et~al.}(2025)\citenamefont {Hassani}, \citenamefont {Scheiner}, \citenamefont {Paris},\ and\ \citenamefont {Markham}}]{Hassani_2025}%
  \BibitemOpen
  \bibfield  {author} {\bibinfo {author} {\bibfnamefont {M.}~\bibnamefont {Hassani}}, \bibinfo {author} {\bibfnamefont {S.}~\bibnamefont {Scheiner}}, \bibinfo {author} {\bibfnamefont {M.~G.}\ \bibnamefont {Paris}},\ and\ \bibinfo {author} {\bibfnamefont {D.}~\bibnamefont {Markham}},\ }\bibfield  {title} {\bibinfo {title} {Privacy in networks of quantum sensors},\ }\bibfield  {journal} {\bibinfo  {journal} {Physical Review Letters}\ }\textbf {\bibinfo {volume} {134}},\ \href {https://doi.org/10.1103/physrevlett.134.030802} {10.1103/physrevlett.134.030802} (\bibinfo {year} {2025})\BibitemShut {NoStop}%
\bibitem [{\citenamefont {Shettell}\ \emph {et~al.}(2022{\natexlab{b}})\citenamefont {Shettell}, \citenamefont {Hassani},\ and\ \citenamefont {Markham}}]{Shettell2022}%
  \BibitemOpen
  \bibfield  {author} {\bibinfo {author} {\bibfnamefont {N.}~\bibnamefont {Shettell}}, \bibinfo {author} {\bibfnamefont {M.}~\bibnamefont {Hassani}},\ and\ \bibinfo {author} {\bibfnamefont {D.}~\bibnamefont {Markham}},\ }\href {http://arxiv.org/abs/2207.14450} {\bibinfo {title} {{Private network parameter estimation with quantum sensors}}} (\bibinfo {year} {2022}{\natexlab{b}})\BibitemShut {NoStop}%
\bibitem [{\citenamefont {Bugalho}\ \emph {et~al.}(2024)\citenamefont {Bugalho}, \citenamefont {Hassani}, \citenamefont {Omar},\ and\ \citenamefont {Markham}}]{Bugalho2024}%
  \BibitemOpen
  \bibfield  {author} {\bibinfo {author} {\bibfnamefont {L.}~\bibnamefont {Bugalho}}, \bibinfo {author} {\bibfnamefont {M.}~\bibnamefont {Hassani}}, \bibinfo {author} {\bibfnamefont {Y.}~\bibnamefont {Omar}},\ and\ \bibinfo {author} {\bibfnamefont {D.}~\bibnamefont {Markham}},\ }\href {http://arxiv.org/abs/2407.21701} {\bibinfo {title} {{Private and Robust States for Distributed Quantum Sensing}}} (\bibinfo {year} {2024}),\ \Eprint {https://arxiv.org/abs/2407.21701} {arXiv:2407.21701} \BibitemShut {NoStop}%
\bibitem [{\citenamefont {Zang}\ \emph {et~al.}(2024)\citenamefont {Zang}, \citenamefont {Kolar}, \citenamefont {Gonzales}, \citenamefont {Chung}, \citenamefont {Gray}, \citenamefont {Kettimuthu}, \citenamefont {Zhong},\ and\ \citenamefont {Saleem}}]{Zang2024}%
  \BibitemOpen
  \bibfield  {author} {\bibinfo {author} {\bibfnamefont {A.}~\bibnamefont {Zang}}, \bibinfo {author} {\bibfnamefont {A.}~\bibnamefont {Kolar}}, \bibinfo {author} {\bibfnamefont {A.}~\bibnamefont {Gonzales}}, \bibinfo {author} {\bibfnamefont {J.}~\bibnamefont {Chung}}, \bibinfo {author} {\bibfnamefont {S.~K.}\ \bibnamefont {Gray}}, \bibinfo {author} {\bibfnamefont {R.}~\bibnamefont {Kettimuthu}}, \bibinfo {author} {\bibfnamefont {T.}~\bibnamefont {Zhong}},\ and\ \bibinfo {author} {\bibfnamefont {Z.~H.}\ \bibnamefont {Saleem}},\ }\href {http://arxiv.org/abs/2409.17089} {\bibinfo {title} {{Quantum Advantage in Distributed Sensing with Noisy Quantum Networks}}} (\bibinfo {year} {2024}),\ \Eprint {https://arxiv.org/abs/2409.17089} {arXiv:2409.17089} \BibitemShut {NoStop}%
\bibitem [{\citenamefont {Li}\ and\ \citenamefont {Smith}(2015)}]{Li2015}%
  \BibitemOpen
  \bibfield  {author} {\bibinfo {author} {\bibfnamefont {K.}~\bibnamefont {Li}}\ and\ \bibinfo {author} {\bibfnamefont {G.}~\bibnamefont {Smith}},\ }\bibfield  {title} {\bibinfo {title} {{Quantum de Finetti Theorem under Fully-One-Way Adaptive Measurements}},\ }\href {https://doi.org/10.1103/PhysRevLett.114.160503} {\bibfield  {journal} {\bibinfo  {journal} {Phys. Rev. Lett.}\ }\textbf {\bibinfo {volume} {114}},\ \bibinfo {pages} {160503} (\bibinfo {year} {2015})}\BibitemShut {NoStop}%
\bibitem [{\citenamefont {Kim}\ \emph {et~al.}(2024)\citenamefont {Kim}, \citenamefont {Hong}, \citenamefont {Kim}, \citenamefont {Kim}, \citenamefont {Lee}, \citenamefont {Pooser}, \citenamefont {Oh}, \citenamefont {Lee}, \citenamefont {Lee},\ and\ \citenamefont {Lim}}]{Kim24}%
  \BibitemOpen
  \bibfield  {author} {\bibinfo {author} {\bibfnamefont {D.-H.}\ \bibnamefont {Kim}}, \bibinfo {author} {\bibfnamefont {S.}~\bibnamefont {Hong}}, \bibinfo {author} {\bibfnamefont {Y.-S.}\ \bibnamefont {Kim}}, \bibinfo {author} {\bibfnamefont {Y.}~\bibnamefont {Kim}}, \bibinfo {author} {\bibfnamefont {S.-W.}\ \bibnamefont {Lee}}, \bibinfo {author} {\bibfnamefont {R.~C.}\ \bibnamefont {Pooser}}, \bibinfo {author} {\bibfnamefont {K.}~\bibnamefont {Oh}}, \bibinfo {author} {\bibfnamefont {S.-Y.}\ \bibnamefont {Lee}}, \bibinfo {author} {\bibfnamefont {C.}~\bibnamefont {Lee}},\ and\ \bibinfo {author} {\bibfnamefont {H.-T.}\ \bibnamefont {Lim}},\ }\bibfield  {title} {\bibinfo {title} {Distributed quantum sensing of multiple phases with fewer photons},\ }\href {https://doi.org/10.1038/s41467-023-44204-z} {\bibfield  {journal} {\bibinfo  {journal} {Nature Communications}\ }\textbf {\bibinfo {volume} {15}},\ \bibinfo {pages} {266} (\bibinfo {year} {2024})}\BibitemShut {NoStop}%
\bibitem [{\citenamefont {Pirandola}\ \emph {et~al.}(2020)\citenamefont {Pirandola}, \citenamefont {Andersen}, \citenamefont {Banchi}, \citenamefont {Berta}, \citenamefont {Bunandar}, \citenamefont {Colbeck}, \citenamefont {Englund}, \citenamefont {Gehring}, \citenamefont {Lupo}, \citenamefont {Ottaviani}, \citenamefont {Pereira}, \citenamefont {Razavi}, \citenamefont {{Shamsul Shaari}}, \citenamefont {Tomamichel}, \citenamefont {Usenko}, \citenamefont {Vallone}, \citenamefont {Villoresi},\ and\ \citenamefont {Wallden}}]{Pirandola2019}%
  \BibitemOpen
  \bibfield  {author} {\bibinfo {author} {\bibfnamefont {S.}~\bibnamefont {Pirandola}}, \bibinfo {author} {\bibfnamefont {U.~L.}\ \bibnamefont {Andersen}}, \bibinfo {author} {\bibfnamefont {L.}~\bibnamefont {Banchi}}, \bibinfo {author} {\bibfnamefont {M.}~\bibnamefont {Berta}}, \bibinfo {author} {\bibfnamefont {D.}~\bibnamefont {Bunandar}}, \bibinfo {author} {\bibfnamefont {R.}~\bibnamefont {Colbeck}}, \bibinfo {author} {\bibfnamefont {D.}~\bibnamefont {Englund}}, \bibinfo {author} {\bibfnamefont {T.}~\bibnamefont {Gehring}}, \bibinfo {author} {\bibfnamefont {C.}~\bibnamefont {Lupo}}, \bibinfo {author} {\bibfnamefont {C.}~\bibnamefont {Ottaviani}}, \bibinfo {author} {\bibfnamefont {J.~L.}\ \bibnamefont {Pereira}}, \bibinfo {author} {\bibfnamefont {M.}~\bibnamefont {Razavi}}, \bibinfo {author} {\bibfnamefont {J.}~\bibnamefont {{Shamsul Shaari}}}, \bibinfo {author} {\bibfnamefont {M.}~\bibnamefont {Tomamichel}}, \bibinfo {author} {\bibfnamefont {V.~C.}\ \bibnamefont {Usenko}}, \bibinfo {author} {\bibfnamefont
  {G.}~\bibnamefont {Vallone}}, \bibinfo {author} {\bibfnamefont {P.}~\bibnamefont {Villoresi}},\ and\ \bibinfo {author} {\bibfnamefont {P.}~\bibnamefont {Wallden}},\ }\bibfield  {title} {\bibinfo {title} {{Advances in quantum cryptography}},\ }\href {https://doi.org/10.1364/AOP.361502} {\bibfield  {journal} {\bibinfo  {journal} {Adv. Opt. Photonics}\ }\textbf {\bibinfo {volume} {12}},\ \bibinfo {pages} {1012} (\bibinfo {year} {2020})},\ \Eprint {https://arxiv.org/abs/1906.01645} {arXiv:1906.01645} \BibitemShut {NoStop}%
\bibitem [{\citenamefont {Shettell}\ and\ \citenamefont {Markham}(2022)}]{Shettell2022b}%
  \BibitemOpen
  \bibfield  {author} {\bibinfo {author} {\bibfnamefont {N.}~\bibnamefont {Shettell}}\ and\ \bibinfo {author} {\bibfnamefont {D.}~\bibnamefont {Markham}},\ }\bibfield  {title} {\bibinfo {title} {{Quantum metrology with delegated tasks}},\ }\href {https://doi.org/10.1103/PhysRevA.106.052427} {\bibfield  {journal} {\bibinfo  {journal} {Phys. Rev. A}\ }\textbf {\bibinfo {volume} {106}},\ \bibinfo {pages} {052427} (\bibinfo {year} {2022})}\BibitemShut {NoStop}%
\bibitem [{\citenamefont {Bilkis}\ \emph {et~al.}(2021)\citenamefont {Bilkis}, \citenamefont {Rosati},\ and\ \citenamefont {Calsamiglia}}]{Bilkis2021}%
  \BibitemOpen
  \bibfield  {author} {\bibinfo {author} {\bibfnamefont {M.}~\bibnamefont {Bilkis}}, \bibinfo {author} {\bibfnamefont {M.}~\bibnamefont {Rosati}},\ and\ \bibinfo {author} {\bibfnamefont {J.}~\bibnamefont {Calsamiglia}},\ }\bibfield  {title} {\bibinfo {title} {{Reinforcement-learning calibration of coherent-state receivers on variable-loss optical channels}},\ }in\ \href {https://doi.org/10.1109/ITW48936.2021.9611396} {\emph {\bibinfo {booktitle} {2021 IEEE Inf. Theory Work.}}}\ (\bibinfo  {publisher} {IEEE},\ \bibinfo {year} {2021})\ pp.\ \bibinfo {pages} {1--6}\BibitemShut {NoStop}%
\bibitem [{\citenamefont {Rosati}(2021)}]{Rosati2021}%
  \BibitemOpen
  \bibfield  {author} {\bibinfo {author} {\bibfnamefont {M.}~\bibnamefont {Rosati}},\ }\bibfield  {title} {\bibinfo {title} {{Performance of Coherent Frequency-Shift Keying for Classical Communication on Quantum Channels}},\ }in\ \href {https://doi.org/10.1109/ISIT45174.2021.9517959} {\emph {\bibinfo {booktitle} {2021 IEEE Int. Symp. Inf. Theory}}}\ (\bibinfo  {publisher} {IEEE},\ \bibinfo {year} {2021})\ pp.\ \bibinfo {pages} {902--905}\BibitemShut {NoStop}%
\bibitem [{\citenamefont {Rosati}\ and\ \citenamefont {Solana}(2024)}]{Rosati2023}%
  \BibitemOpen
  \bibfield  {author} {\bibinfo {author} {\bibfnamefont {M.}~\bibnamefont {Rosati}}\ and\ \bibinfo {author} {\bibfnamefont {A.}~\bibnamefont {Solana}},\ }\bibfield  {title} {\bibinfo {title} {{Joint-detection learning for optical communication at the quantum limit}},\ }\href {https://doi.org/10.1364/OPTICAQ.521637} {\bibfield  {journal} {\bibinfo  {journal} {Opt. Quantum}\ }\textbf {\bibinfo {volume} {2}},\ \bibinfo {pages} {390} (\bibinfo {year} {2024})},\ \Eprint {https://arxiv.org/abs/2312.13693} {arXiv:2312.13693} \BibitemShut {NoStop}%
\bibitem [{\citenamefont {N{\"{o}}tzel}\ and\ \citenamefont {Rosati}(2023)}]{Notzel2022}%
  \BibitemOpen
  \bibfield  {author} {\bibinfo {author} {\bibfnamefont {J.}~\bibnamefont {N{\"{o}}tzel}}\ and\ \bibinfo {author} {\bibfnamefont {M.}~\bibnamefont {Rosati}},\ }\bibfield  {title} {\bibinfo {title} {{Operating Fiber Networks in the Quantum Limit}},\ }\href {https://doi.org/10.1109/JLT.2023.3295076} {\bibfield  {journal} {\bibinfo  {journal} {J. Light. Technol.}\ }\textbf {\bibinfo {volume} {41}},\ \bibinfo {pages} {6865} (\bibinfo {year} {2023})},\ \Eprint {https://arxiv.org/abs/2201.12397} {arXiv:2201.12397} \BibitemShut {NoStop}%
\bibitem [{\citenamefont {Munar-Vallespir}\ and\ \citenamefont {N{\"{o}}tzel}(2024)}]{Munar-Vallespir2024}%
  \BibitemOpen
  \bibfield  {author} {\bibinfo {author} {\bibfnamefont {P.}~\bibnamefont {Munar-Vallespir}}\ and\ \bibinfo {author} {\bibfnamefont {J.}~\bibnamefont {N{\"{o}}tzel}},\ }\bibfield  {title} {\bibinfo {title} {{Joint communication and sensing over the lossy bosonic quantum channel}},\ }in\ \href {https://doi.org/10.1109/WF-IoT62078.2024.10811180} {\emph {\bibinfo {booktitle} {2024 IEEE 10th World Forum Internet Things}}}\ (\bibinfo  {publisher} {IEEE},\ \bibinfo {year} {2024})\ pp.\ \bibinfo {pages} {1--6}\BibitemShut {NoStop}%
\bibitem [{\citenamefont {Renner}(2007)}]{Renner2007}%
  \BibitemOpen
  \bibfield  {author} {\bibinfo {author} {\bibfnamefont {R.}~\bibnamefont {Renner}},\ }\bibfield  {title} {\bibinfo {title} {{Symmetry of large physical systems implies independence of subsystems}},\ }\href {https://doi.org/10.1038/nphys684} {\bibfield  {journal} {\bibinfo  {journal} {Nat. Phys.}\ }\textbf {\bibinfo {volume} {3}},\ \bibinfo {pages} {645} (\bibinfo {year} {2007})},\ \Eprint {https://arxiv.org/abs/0703069} {arXiv:0703069 [quant-ph]} \BibitemShut {NoStop}%
\bibitem [{\citenamefont {Brandao}\ and\ \citenamefont {Harrow}(2013)}]{Brandao2013}%
  \BibitemOpen
  \bibfield  {author} {\bibinfo {author} {\bibfnamefont {F.~G.}\ \bibnamefont {Brandao}}\ and\ \bibinfo {author} {\bibfnamefont {A.~W.}\ \bibnamefont {Harrow}},\ }\bibfield  {title} {\bibinfo {title} {{Quantum de finetti theorems under local measurements with applications}},\ }in\ \href {https://doi.org/10.1145/2488608.2488718} {\emph {\bibinfo {booktitle} {Proc. forty-fifth Annu. ACM Symp. Theory Comput.}}}\ (\bibinfo  {publisher} {ACM},\ \bibinfo {address} {New York, NY, USA},\ \bibinfo {year} {2013})\ pp.\ \bibinfo {pages} {861--870}\BibitemShut {NoStop}%
\bibitem [{\citenamefont {Winter}(1999)}]{winter1999}%
  \BibitemOpen
  \bibfield  {author} {\bibinfo {author} {\bibfnamefont {A.}~\bibnamefont {Winter}},\ }\bibfield  {title} {\bibinfo {title} {{Coding theorem and strong converse for quantum channels}},\ }\href {https://doi.org/10.1109/18.796385} {\bibfield  {journal} {\bibinfo  {journal} {IEEE Trans. Inf. Theory}\ }\textbf {\bibinfo {volume} {45}},\ \bibinfo {pages} {2481} (\bibinfo {year} {1999})},\ \Eprint {https://arxiv.org/abs/1409.2536} {arXiv:1409.2536} \BibitemShut {NoStop}%
\bibitem [{\citenamefont {Langford}\ \emph {et~al.}(2005)\citenamefont {Langford}, \citenamefont {Weinhold}, \citenamefont {Prevedel}, \citenamefont {Resch}, \citenamefont {Gilchrist}, \citenamefont {O'Brien}, \citenamefont {Pryde},\ and\ \citenamefont {White}}]{Langford05}%
  \BibitemOpen
  \bibfield  {author} {\bibinfo {author} {\bibfnamefont {N.~K.}\ \bibnamefont {Langford}}, \bibinfo {author} {\bibfnamefont {T.~J.}\ \bibnamefont {Weinhold}}, \bibinfo {author} {\bibfnamefont {R.}~\bibnamefont {Prevedel}}, \bibinfo {author} {\bibfnamefont {K.~J.}\ \bibnamefont {Resch}}, \bibinfo {author} {\bibfnamefont {A.}~\bibnamefont {Gilchrist}}, \bibinfo {author} {\bibfnamefont {J.~L.}\ \bibnamefont {O'Brien}}, \bibinfo {author} {\bibfnamefont {G.~J.}\ \bibnamefont {Pryde}},\ and\ \bibinfo {author} {\bibfnamefont {A.~G.}\ \bibnamefont {White}},\ }\bibfield  {title} {\bibinfo {title} {Demonstration of a simple entangling optical gate and its use in bell-state analysis},\ }\href {https://doi.org/10.1103/PhysRevLett.95.210504} {\bibfield  {journal} {\bibinfo  {journal} {Phys. Rev. Lett.}\ }\textbf {\bibinfo {volume} {95}},\ \bibinfo {pages} {210504} (\bibinfo {year} {2005})}\BibitemShut {NoStop}%
\bibitem [{\citenamefont {Kiesel}\ \emph {et~al.}(2005)\citenamefont {Kiesel}, \citenamefont {Schmid}, \citenamefont {Weber}, \citenamefont {Ursin},\ and\ \citenamefont {Weinfurter}}]{Kiesel05}%
  \BibitemOpen
  \bibfield  {author} {\bibinfo {author} {\bibfnamefont {N.}~\bibnamefont {Kiesel}}, \bibinfo {author} {\bibfnamefont {C.}~\bibnamefont {Schmid}}, \bibinfo {author} {\bibfnamefont {U.}~\bibnamefont {Weber}}, \bibinfo {author} {\bibfnamefont {R.}~\bibnamefont {Ursin}},\ and\ \bibinfo {author} {\bibfnamefont {H.}~\bibnamefont {Weinfurter}},\ }\bibfield  {title} {\bibinfo {title} {Linear optics controlled-phase gate made simple},\ }\href {https://doi.org/10.1103/PhysRevLett.95.210505} {\bibfield  {journal} {\bibinfo  {journal} {Phys. Rev. Lett.}\ }\textbf {\bibinfo {volume} {95}},\ \bibinfo {pages} {210505} (\bibinfo {year} {2005})}\BibitemShut {NoStop}%
\bibitem [{\citenamefont {Okamoto}\ \emph {et~al.}(2005)\citenamefont {Okamoto}, \citenamefont {Hofmann}, \citenamefont {Takeuchi},\ and\ \citenamefont {Sasaki}}]{Okamoto05}%
  \BibitemOpen
  \bibfield  {author} {\bibinfo {author} {\bibfnamefont {R.}~\bibnamefont {Okamoto}}, \bibinfo {author} {\bibfnamefont {H.~F.}\ \bibnamefont {Hofmann}}, \bibinfo {author} {\bibfnamefont {S.}~\bibnamefont {Takeuchi}},\ and\ \bibinfo {author} {\bibfnamefont {K.}~\bibnamefont {Sasaki}},\ }\bibfield  {title} {\bibinfo {title} {Demonstration of an optical quantum controlled-not gate without path interference},\ }\href {https://doi.org/10.1103/PhysRevLett.95.210506} {\bibfield  {journal} {\bibinfo  {journal} {Phys. Rev. Lett.}\ }\textbf {\bibinfo {volume} {95}},\ \bibinfo {pages} {210506} (\bibinfo {year} {2005})}\BibitemShut {NoStop}%
\bibitem [{\citenamefont {Bizzarri}\ \emph {et~al.}(2024{\natexlab{a}})\citenamefont {Bizzarri}, \citenamefont {Gianani}, \citenamefont {Manrique}, \citenamefont {Berardi}, \citenamefont {Bruni}, \citenamefont {Capellini},\ and\ \citenamefont {Barbieri}}]{bizzarri2024steering}%
  \BibitemOpen
  \bibfield  {author} {\bibinfo {author} {\bibfnamefont {G.}~\bibnamefont {Bizzarri}}, \bibinfo {author} {\bibfnamefont {I.}~\bibnamefont {Gianani}}, \bibinfo {author} {\bibfnamefont {M.}~\bibnamefont {Manrique}}, \bibinfo {author} {\bibfnamefont {V.}~\bibnamefont {Berardi}}, \bibinfo {author} {\bibfnamefont {F.}~\bibnamefont {Bruni}}, \bibinfo {author} {\bibfnamefont {G.}~\bibnamefont {Capellini}},\ and\ \bibinfo {author} {\bibfnamefont {M.}~\bibnamefont {Barbieri}},\ }\bibfield  {title} {\bibinfo {title} {Quantum steering from phase measurements with limited resources},\ }\href {https://doi.org/10.1116/5.0205469} {\bibfield  {journal} {\bibinfo  {journal} {AVS Quantum Science}\ }\textbf {\bibinfo {volume} {6}},\ \bibinfo {pages} {024405} (\bibinfo {year} {2024}{\natexlab{a}})},\ \Eprint {https://arxiv.org/abs/https://pubs.aip.org/avs/aqs/article-pdf/doi/10.1116/5.0205469/19956771/024405\_1\_5.0205469.pdf} {https://pubs.aip.org/avs/aqs/article-pdf/doi/10.1116/5.0205469/19956771/024405\_1\_5.0205469.pdf}
  \BibitemShut {NoStop}%
\bibitem [{\citenamefont {Bizzarri}\ \emph {et~al.}(2024{\natexlab{b}})\citenamefont {Bizzarri}, \citenamefont {Gherardini}, \citenamefont {Manrique}, \citenamefont {Bruni}, \citenamefont {Gianani},\ and\ \citenamefont {Barbieri}}]{bizzarri2024quasiprob}%
  \BibitemOpen
  \bibfield  {author} {\bibinfo {author} {\bibfnamefont {G.}~\bibnamefont {Bizzarri}}, \bibinfo {author} {\bibfnamefont {S.}~\bibnamefont {Gherardini}}, \bibinfo {author} {\bibfnamefont {M.}~\bibnamefont {Manrique}}, \bibinfo {author} {\bibfnamefont {F.}~\bibnamefont {Bruni}}, \bibinfo {author} {\bibfnamefont {I.}~\bibnamefont {Gianani}},\ and\ \bibinfo {author} {\bibfnamefont {M.}~\bibnamefont {Barbieri}},\ }\href {https://arxiv.org/abs/2406.06713} {\bibinfo {title} {Quasiprobability distributions with weak measurements}} (\bibinfo {year} {2024}{\natexlab{b}}),\ \Eprint {https://arxiv.org/abs/2406.06713} {arXiv:2406.06713 [quant-ph]} \BibitemShut {NoStop}%
\bibitem [{\citenamefont {Bizzarri}\ \emph {et~al.}(2025)\citenamefont {Bizzarri}, \citenamefont {Parisi}, \citenamefont {Manrique}, \citenamefont {Gianani}, \citenamefont {Chiuri}, \citenamefont {Rosati}, \citenamefont {Giovannetti}, \citenamefont {Paris},\ and\ \citenamefont {Barbieri}}]{bizzarri2025sloppy}%
  \BibitemOpen
  \bibfield  {author} {\bibinfo {author} {\bibfnamefont {G.}~\bibnamefont {Bizzarri}}, \bibinfo {author} {\bibfnamefont {M.}~\bibnamefont {Parisi}}, \bibinfo {author} {\bibfnamefont {M.}~\bibnamefont {Manrique}}, \bibinfo {author} {\bibfnamefont {I.}~\bibnamefont {Gianani}}, \bibinfo {author} {\bibfnamefont {A.}~\bibnamefont {Chiuri}}, \bibinfo {author} {\bibfnamefont {M.}~\bibnamefont {Rosati}}, \bibinfo {author} {\bibfnamefont {V.}~\bibnamefont {Giovannetti}}, \bibinfo {author} {\bibfnamefont {M.~G.~A.}\ \bibnamefont {Paris}},\ and\ \bibinfo {author} {\bibfnamefont {M.}~\bibnamefont {Barbieri}},\ }\href {https://arxiv.org/abs/2503.04976} {\bibinfo {title} {Retrieving information in sloppy quantum phase estimation}} (\bibinfo {year} {2025}),\ \Eprint {https://arxiv.org/abs/2503.04976} {arXiv:2503.04976 [quant-ph]} \BibitemShut {NoStop}%
\bibitem [{Rep()}]{Repo}%
  \BibitemOpen
  \href@noop {} {\bibinfo {title} {10.5281/zenodo.15271379}}\BibitemShut {NoStop}%
\bibitem [{\citenamefont {Ac{\'{i}}n}(2001)}]{Acin2001a}%
  \BibitemOpen
  \bibfield  {author} {\bibinfo {author} {\bibfnamefont {A.}~\bibnamefont {Ac{\'{i}}n}},\ }\bibfield  {title} {\bibinfo {title} {{Statistical Distinguishability between Unitary Operations}},\ }\href {https://doi.org/10.1103/PhysRevLett.87.177901} {\bibfield  {journal} {\bibinfo  {journal} {Phys. Rev. Lett.}\ }\textbf {\bibinfo {volume} {87}},\ \bibinfo {pages} {177901} (\bibinfo {year} {2001})},\ \Eprint {https://arxiv.org/abs/0102064} {arXiv:0102064 [quant-ph]} \BibitemShut {NoStop}%
\end{thebibliography}%

\appendix
\onecolumngrid

\section{Equivalence of entanglement- and MUB-based protocols}\label{app:equivalence}
\begin{proof} (Proof of Lemma~\ref{lemma:equivalence}
)\\
In the entanglement-based protocol, Alice's register never interacts with Eve, hence the expectations in \eqref{eq:fidelity_ent} can be expressed as follows:
\begin{align}
    &\ave{Z\otimes \bX} = \frac{\tr{\bX \cdot\Phi(\dketbra{\boldsymbol{+}})} - \tr{\bX \cdot\Phi(\dketbra{\boldsymbol{-}})}}{2},\label{eq:zx_mean}\\
    &\ave{X\otimes \bZ} = \frac{\tr{\bZ \cdot\Phi(\dketbra{\bf0})} - \tr{\bZ \cdot\Phi(\dketbra{\bf1})}}{2},\label{eq:xz_mean}\\
    &\ave{Y\otimes \bY} = \frac{\tr{\bY \cdot\Phi(\dketbra{\bf R})} - \tr{\bY \cdot\Phi(\dketbra{\bf L})}}{2},
\end{align}
where we have used that
\begin{equation}
    \ket{\boldsymbol{\psi_+}}\equiv\frac{\ket{+}\ket{\bf0}+\ket{-}\ket{\bf1}}{\sqrt2}\equiv\frac{\ket R\ket{\bf R} + \ket L \ket{\bf L}}{\sqrt2}
\end{equation}
and the fact that channel controlled by Eve acts only on the $B$-register, leaving the $A$-register invariant. We identify in the above expressions differences of expectations of Pauli-like operators on their evolved $+1$-eigenstates under $\Phi$. Furthermore, we notice that, for any $\bP$, 
\begin{equation}
    \ave{\bP}=\tr{\bP \cdot\Phi(\dketbra{\bP,+})} = 2 F(\ket{\bP,+},\Phi(\dketbra{\bP,+}))-1 \equiv 2\hat F_\bP-1.
\end{equation}
Plugging these into~\eqref{eq:fidelity_ent}, we obtain
\begin{align}
    \hat F = \frac14\left(1+\frac12\sum_{\bP}(2 \hat F_\bP-1)\right).
\end{align}
Therefore, if $\hat F_\bP \geq 1-\bar \epsilon^2$ for all $\bP$ in Protocol~\ref{protocol:1w_mub}, then in Protocol~\ref{protocol:1w_ent} it holds $\hat F \geq 1-\epsilon^2$, with $\epsilon^2 = 3\bar\epsilon^2/2$. Furthermore, using \eqref{eq:zx_mean},\eqref{eq:xz_mean} it is straightforward to show that the estimation step calculates the very same quantity in both protocols.
The expected number of rounds in Protocol~\ref{protocol:1w_mub} is slightly different, considering that no sifting is needed:  we have $N_e =  p_e T$, $N_c = p_c T$ and $N_d = p_d T$ which can be expressed in terms of $T_{e,c,d}$.
\end{proof}

\section{Proofs for one-way faithfulness under GC attacks}\label{app:proofs}
\begin{proof}(Proof of Lemma~\ref{lemma:gentle_measurement})\\
    Using \eqref{eq:def_locc_distance} and the variational formulation of trace-distance we have
    \begin{align}
        2D_{\locc}(\tau,\sigma) =  \max_{\substack{\Lambda\text{  }\locc,\\ \1\geq E\geq0}} \Tr[E \cdot \Lambda(\tau-\sigma)] \geq \Tr[E\tau]-\Tr[E\sigma],
    \end{align}
    where the inequality follows by choosing a specific $E$ and $\Lambda={\rm Id}$. The result follows by rearranging terms, and one can easily switch the roles of $\tau$ and $\sigma$.
\end{proof}
\begin{proof}
(Proof of Theorem~\ref{thm:estimation_under_locc_tampering})\\
To prove the bound for the expectation, we apply Lemma \ref{lemma:unif_cont} with
$A =\bar O$. This implies $a = o/m$ and $K=m$, and we obtain
    \begin{align}
       \left|\E_{\sigma} \bar O-\E_{\sigma'} \bar O\right|  =\left|\Tr[A(\sigma-\sigma')]\right| &\leq 2o\epsilon.
    \end{align}
    To prove the inequality for the variance, following~\cite{Shettell2022a}, we first note that
    \begin{align}
        \Delta^2_{\sigma,\sigma'}\bar O &= \Tr[(\bar O-\E_{\sigma'} \bar O+\E_{\sigma'} \bar O-\E_{\sigma} \bar O)^2 \sigma'] \\
        &= {\rm Var}_{\sigma'}\bar O + \left(\E_{\sigma'} \bar O - \E_{\sigma} \bar O\right)^2
    \end{align}
with ${\rm Var}_{\sigma}\bar O = \Delta^2_{\sigma,\sigma}\bar O $, and hence
    \begin{equation}
        \left|\Delta^2_{\sigma,\sigma} \bar O-\Delta^2_{\sigma,\sigma'} \bar O\right| \leq \left|{\rm Var}_{\sigma}\bar O - {\rm Var}_{\sigma'}\bar O\right| + \left(\E_{\sigma'} \bar O - \E_{\sigma}\bar O\right)^2.
\end{equation}
    The second term above can be bounded by $(2o\epsilon)^2$ via \eqref{eq:bias_bound_gc}. For the first one, instead, we can write
    \begin{equation}
        {\rm Var_\sigma \bar O} = \Tr[O^2\sigma]-(\Tr[O\sigma])^2 = \Tr[(O^2\otimes \1 - O^{\otimes 2})\sigma^{\otimes 2}].
    \end{equation}
    Therefore, we can apply Lemma~\ref{lemma:unif_cont} with $A=\frac1m\sum_{j=1}^{m} (O^2\otimes \1 - O^{\otimes 2})_j\otimes\1^{\otimes(m-1)}$ on two copies of the states, satisfying $D_{\locc}(\sigma^{\otimes 2},\sigma'^{\otimes 2})\leq 2\epsilon$. This implies $a=2o^2/m$ and $K=m$, and we obtain
    \begin{align}
        \left|\Delta^2_{\sigma,\sigma} \bar O-\Delta^2_{\sigma,\sigma'} \bar O\right|&\leq\Tr\left[A (\sigma^{\otimes 2}-\sigma'^{\otimes 2})\right] + 4o^2\epsilon^2\\
        &\leq 4o^2 \cdot 2\epsilon + 4o^2\epsilon^2 = 4o^2(2\epsilon+\epsilon^2).
    \end{align}
\end{proof}

\section{Proof of faithfulness in the two-way MUB-based protocol}\label{app:2w}
Here we consider the extension of our one-way protocols to a two-way scenario. As already noted, the form of both Protocols~\ref{protocol:1w_ent} and~\ref{protocol:1w_mub} has very minor changes, essentially with Bob sending Alice the quantum state after the Encoding step and Alice replacing Bob's role in all following steps. 

As also noted above, the security of the protocol cannot be guaranteed in this case. Indeed, while one could imagine to have Bob disclose the estimation rounds only after Alice has checked for the presence of Eve, this latter task is harder than it appears. In particular, it is hard for Alice to quantify, from her received state, how much information might have leaked to Eve. 

In the following we thus focus on faithfulness under GC attacks, which can be guaranteed also in this case. 
\begin{proof}
(Proof of the two-way part of Theorem~\ref{thm:main_mub})    
With respect to the one-way protocol, in the two-way case the major difficulty is that Eve can interact with the quantum state both before and after Bob, possibly even employing a quantum memory. This makes the estimation of tampering on the estimation rounds from the check rounds more difficult, changing Theorem~\ref{thm:distance_bound_1w_mub}. We refer to the communication rounds from Alice to Bob as \emph{forth} communication and to those from Bob to Alice as \emph{back} communication. After $T$ rounds, and before any measurement takes place, the quantum state held by Alice can be written as
\begin{equation}
    \tilde\tau(\phi)= \Tr_E[W\, (V(\phi)\otimes \1)\tilde\sigma_{BE}(V(\phi)\otimes\1) W^\dagger],
\end{equation}
where $\tilde \sigma_{BE}$ is the state shared by Bob and Eve after $T$ transmissions from Alice to Bob and $W$ is a unitary dilation of Eve's back channel acting on the states that Bob sends to Alice. 
Note that also in this case the state $\tilde\sigma$ can be taken permutation-invariant, as the protocol still employs randomized measurements. 

We can then bound the $\locc$-distance to the ideal phase-encoded state similarly to the one-way case:
\begin{align}
    &D_{\locc}(V(\phi)\ket{\psi}^{\otimes N_e},\Tr_{T-N_e}[\tilde\tau(\phi)]) \\
&\leq D_{\locc}\left(V(\phi)\ket{\psi}^{\otimes N_e},\int d\mu(\rho)  \rho^{\otimes N_e}\right) + f(T,N,n).
 \label{eq:decomposition_tampering_two_way}
\end{align}
where the inequality follows from Theorem~\ref{thm:de_finetti}, with the same quantities as defined for \eqref{eq:de_finetti_applied} and the second argument of the $\locc$-distance can be $\phi$-dependent. The first term above can be treated as follows:
\begin{align}
    &D_{\locc}\left(V(\phi)\ket{\psi}^{\otimes N_e},\int d\mu(\rho)  \rho^{\otimes N_e}\right) \\
    & \leq D_{\locc}\left(V(\phi)\ket{\psi}^{\otimes N_e},\int d\mu(\rho)  V(\phi)\rho^{\otimes N_e} V(\phi)^\dagger\right)\\
    &+D\left(\int d\mu(\rho)  V(\phi)\rho^{\otimes N_e} V(\phi)^\dagger,\int d\mu(\rho)  \rho^{\otimes N_e}\right) \\
    &\leq \sqrt{1- \int d\mu(\rho) F(\dketbra{\psi},\rho)^{N_e} } + \int d\mu(\rho) D(V(\phi)\rho^{\otimes N_e} V(\phi)^\dagger,\rho^{\otimes N_e}) \\
    &\leq \sqrt{1- F((\dketbra{\psi})^{\otimes N_c},\Tr_{T-N_c}[\tilde\sigma_B]) + 2 f(T,N,n)} + \int d\mu(\rho) D(V(\phi)\rho^{\otimes N_e} V(\phi)^\dagger,\rho^{\otimes N_e}) ,\label{eq:2w_intermediate_step_distance_bound}
\end{align}
where, in the first term of the last inequality, we have used \eqref{eq:gentle_check} and the fact that $x^{N_e}\geq x^{N_c}$ for $1\geq x\geq 0$ and $N_c\geq N_e$, and $\tilde\sigma_B = \Tr_E[\tilde\sigma_{BE}]$. 

For the second term in \eqref{eq:2w_intermediate_step_distance_bound}, instead, we need to bound the  distance between a product state $\rho^{N_c}$ and its rotated version under $V(\phi)$, proceeding as follows:
\begin{align}
    D(\rho^{N_c},V(\phi)\rho^{N_c} V(\phi)) &\leq\sqrt{1- F(\rho,U(\phi)\rho U(\phi)^\dagger)^{N_e}}\leq\sqrt{1- F(\ket v,(U(\phi)\otimes \1)\ket v )^{N_e}}
\end{align}
where in the first equality we have used the Fuchs-van-de-Graf inequality and recalled that $V(\phi)=U(\phi)^{\otimes N_e}$, while in the second one the non-increasing property of trace-distance under the partial-trace channel and $\ket v$ is a purification of $\rho$. Finally, taking the worst-case over all possible states $\ket v$, we obtain
\begin{align}
    D(\rho^{N_c},V(\phi)\rho^{N_c} V(\phi)) &\leq  \sqrt{1- \min_{\ket v} F(\ket v,(U(\phi)\otimes \1)\ket v )^{N_e}},
\end{align}
which corresponds to Eve being able to discriminate the presence or absence of $U(\phi)$ with the largest success probability allowed by quantum mechanics. This problem has been studied by Acin~\cite{Acin2001a}, who proved that
\begin{equation}
    \min_{\ket v} F(\ket v,U(\phi)\ket v) = \cos^2\min\left\{\frac{\delta(U(\phi))}{2},\frac\pi2\right\},
\end{equation}
where $\delta$ is the minimum arc-length on the unit-circle that contains all eigenvalues of $U$. Putting all together we obtain the bound{\small
\begin{equation}\label{eq:2w_distance_bound}
    D_{\locc}\left(V(\phi)\ket{\psi}^{\otimes N_e},\int d\mu(\rho)  \rho^{\otimes N_e}\right) \leq \sqrt{1- F((\dketbra{\psi})^{\otimes N_c},\Tr_{T-N_c}[\tilde\sigma_B]) + 2 f(T,N,n)} + \left|\sin\frac{\min\{\delta(U(\phi)),\pi\}}{2}\right|
\end{equation}
}
Note that, in our case, 
\begin{equation}
    U(\phi) = (e^{-\ii \phi}\dketbra R + e^{\ii \phi}\dketbra L)^{\otimes n}
\end{equation}
has eigenvalues $\lambda_j = e^{-\ii n\phi (1-2j)}$ for $j=0,\cdots,n$. Hence $\delta(U(\phi)) = 2 n \phi$ for $n\phi<\pi$. 
Finally, the distance bound~\eqref{eq:2w_distance_bound} can be employed as in the one-way case to obtain the theorem statement.
\end{proof}

\end{document}